\documentclass[11pt,letterpaper]{article}
\usepackage{graphicx} 
\usepackage{amsmath,amssymb}
\usepackage{url}
\usepackage{latexsym}

\usepackage{tikz}
\usetikzlibrary{decorations.pathmorphing, decorations.pathreplacing, decorations.shapes}

\usepackage[utf8]{inputenc}
\usepackage[T1]{fontenc}
\usepackage[letterpaper, margin=2.54cm]{geometry}
\usepackage{needspace}
\usepackage{bbm}

\usepackage{libertine} 
\usepackage{inconsolata} 

\usepackage{amsmath, amsthm, amssymb}
\usepackage{dsfont}
\usepackage[dvipsnames]{xcolor}
\usepackage{graphicx}
\usepackage{enumerate}
\usepackage{authblk}
\usepackage{thm-restate}
\usepackage{thmtools}
\usepackage{float}
\usepackage{amsfonts,mathtools,enumerate}
\usepackage[ruled,vlined]{algorithm2e}

\usepackage[colorlinks=true, citecolor=red, pagebackref]{hyperref}
\usepackage[capitalize]{cleveref}
\usepackage{url} 

\usepackage{tikz}
\usetikzlibrary{arrows,decorations,backgrounds} 

\renewenvironment{abstract}
{\small\vspace{-1em}
\begin{center}
\bfseries\abstractname\vspace{-.5em}\vspace{0pt}
\end{center}
\list{}{
\setlength{\leftmargin}{0.6in}%
\setlength{\rightmargin}{\leftmargin}}%
\item\relax}
{\endlist}

\declaretheorem[name=Theorem, numberwithin=section]{theorem}

\declaretheorem[name=Lemma, sibling=theorem]{lemma}
\declaretheorem[name=Proposition, sibling=theorem]{proposition}

\declaretheorem[name=Corollary, sibling=theorem]{corollary}
\declaretheorem[name=Conjecture, sibling=theorem]{conjecture}
\declaretheorem[name=Problem, sibling=theorem]{problem}
\declaretheorem[name=Claim, sibling=theorem]{claim}

\declaretheorem[name=Observation,sibling=theorem]{observation}

\def\cqedsymbol{\ifmmode$\lrcorner$\else{\unskip\nobreak\hfil
\penalty50\hskip1em\null\nobreak\hfil$\lrcorner$
\parfillskip=0pt\finalhyphendemerits=0\endgraf}\fi}

\let\h\H
\def\Pr{\mathbb{P}}
\def\Ex{\mathbb{E}}
\def\D{\mathcal{D}} 
\def\H{\mathcal{H}} 

\newcommand{\NN}{\mathbb{N}} 
\def\O{\mathcal{O}} 

\tikzstyle{vertex}=[circle, draw, fill=black!50,
                        inner sep=0pt, minimum width=4pt]

\let\le\leqslant
\let\ge\geqslant
\let\leq\leqslant
\let\geq\geqslant

\newcommand{\ora}[1]{\overrightarrow{#1}}

\newcommand{\wt}{w}
\DeclareMathOperator{\TEXP}{T_{EXP}}

\newenvironment{proofclaim}[1][]
    {\let\oldqed\qedsymbol\renewcommand{\qedsymbol}{\ensuremath{\lozenge}}\begin{proof}[Proof of claim] }{\end{proof}\renewcommand{\qedsymbol}{\oldqed}}

\renewcommand*{\backref}[1]{}
\renewcommand*{\backrefalt}[4]{
	\ifcase #1 Not cited.%
	\or $\uparrow$#2%
	\else $\uparrow$#2%
	\fi%
}

\newcommand{\defin}[1]{\emph{\textcolor{ForestGreen}{#1}}}

\usepackage{amsthm}

\DeclarePairedDelimiter{\abs}{\lvert}{\rvert}
\DeclarePairedDelimiter{\floor}{\lfloor}{\rfloor}
\let\emptyset\varnothing

\newcommand{\dmax}{d_{\max}}

\usepackage{authblk} 
\newcommand\emailfont{\sffamily}

\title{Improved exploration of temporal graphs}

\author[1]{Paul Bastide\thanks{Research supported by ERC Advanced Grant 883810.}}
\author[2]{Carla Groenland\thanks{Research supported by the Dutch Research Council (NWO, VI.Veni.232.073).}}
\author[3]{Lukas Michel}
\author[4]{Clément Rambaud}
\author[5]{Paulina Smolarova}
\affil[1,3]{
Mathematical Institute, University of Oxford, UK. Emails: {\emailfont{paul.bastide@ens-rennes.fr}, \emailfont{lukas.michel@maths.ox.ac.uk}}.}
\affil[2]{
Delft Institute of Applied Mathematics, TU Delft, the Netherlands. Email: 
{\emailfont{c.e.groenland@tudelft.nl}}.}
\affil[4]{Université Côte d'Azur, CNRS, Inria, I3S, Sophia-Antipolis, France. Email: 
{\emailfont{clement.rambaud@normalesup.org}}.}
\affil[5]{Department of Computer Science, University of Oxford, UK. Email: \emailfont{paulina.smolarova@cs.ox.ac.uk}}
\date{}

\begin{document}

\maketitle

\begin{abstract}
    A temporal graph $G$ is a sequence $(G_t)_{t \in I}$ of graphs on the same vertex set of size $n$.
    The \emph{temporal exploration problem} asks for the length of the shortest sequence of vertices that starts at a given vertex, visits every vertex, and at each time step $t$ either stays at the current vertex or moves to an adjacent vertex in $G_t$.
    Bounds on the length of a shortest temporal exploration have been investigated extensively. Perhaps the most fundamental case is when each graph $G_t$ is connected and has bounded maximum degree. In this setting, Erlebach, Kammer, Luo, Sajenko, and Spooner [ICALP 2019] showed that there exists an exploration of $G$ in $\mathcal{O}(n^{7/4})$ time steps. 
    
    We significantly improve upon this bound by showing that $\mathcal{O}(n^{3/2} \sqrt{\log n})$ time steps suffice.
    In fact, we prove a much more general statement that handles this setting together with many others in a unified manner.
    Let the \emph{average temporal maximum degree} $D$ of $G$ be the average of $\max_{t \in I} d_{G_t}(v)$ over all vertices $v \in V(G)$, where $d_{G_t}(v)$ denotes the degree of $v$ in $G_t$.
    We show that if each graph $G_t$ is connected, then there exists an exploration of $G$ in $\mathcal{O}(n^{3/2} \sqrt{D \log n})$ time steps.
    In particular, this gives the first subquadratic upper bound when the underlying graph has bounded average degree. As a special case, this also improves the previous best bounds when the underlying graph is planar or has bounded treewidth. 
    Our bound is subquadratic already when $D=o(n/\log n)$.

    In addition, we make progress on a recent conjecture from Baguley, G\"{o}bel, Klodt, Skretas, Sylvester, and Zamaraev~[SODA 2025]. They conjectured that if $\mu$ is a probability distribution on trees with maximum degree at most $\Delta$ and each graph $G_t$ is sampled independently from $\mu$, then there exists an exploration in $\mathcal{O}(\sqrt{\Delta} n)$ time steps with high probability. We show that there is a temporal exploration of length $\mathcal{O}(\Delta n \log^3 n)$ with high probability. 
    This improves the best previous bound of $\mathcal{O}(n^{3/2})$ and also extends to the more general setup of average temporal maximum degree.
\end{abstract}

\section{Introduction}

A great variety of topics, ranging from communication networks to the design of power grids or the study of metabolism, can be represented via graphs \cite{newman2018networks}. For many of these  applications, the underlying graph evolves with time, with some connections appearing or disappearing. The dynamics of these changes is fundamental to the study of such problems and is not captured by static networks. To address this issue, temporal graphs were introduced, as they provide a framework for modelling networks that change over time.
There has been recent interest in studying complexity-theoretic problems on temporal graphs~\cite{complexityEulerian23,complexityseparator20,complexityvertexcover22,complexityEulerian21,complexitycoloring22}, specifically related to  reachability questions~\cite{complexityteamofagents14,complexityexplorationstars21,complexityexplorationkernelizing23,pathwidthtwostillhard19,ErlebachHoffmannKammer21,complexitynonstrict20,complexityErlebachSpooner23,MichailSpirakis16}, as well as in proving combinatorial bounds for such questions~\cite{AdamsonGusevMalyshevZamaraev22,BaguleyGobelKlodtSkretasSylvesterZamaraevRST25,specialcases25minedges,DogeasErlebachKammerMeintrupMoses24,ErlebachHoffmannKammer21,ErlebachSpooner18,specialcases22,tempexpcacti14, specialcases13,taghian2020exploring}.
We refer the reader to \cite{casteigts2012time,michail2016introduction} for surveys on temporal graphs.

Formally, a temporal graph $G$ over an interval $I \subseteq \NN$ is a sequence of graphs $(G_t)_{t \in I}$ that all have the same vertex set $V$, but possibly different edge sets. 
We call $G_t$ the \emph{snapshot} at time step $t$, we say that $(V,\bigcup_{t\in I} E(G_t))$ is the \textit{underlying graph} of $G$, and we say that $G$ is \emph{always-connected} if $G_t$ is connected for all $t$. 

In the \textit{temporal exploration problem} (TEXP), introduced by Michail and Spirakis~\cite{MichailSpirakis16}, an agent starts at a given vertex of a temporal graph $G$ and wants to visit all vertices of $V$ as quickly as possible. 
At each time step $t$, the agent can either stay at the current vertex or use an edge of the snapshot $G_t$ to move to an adjacent vertex. We assume that the agent knows the sequence of graphs in advance. 
For a temporal graph $G = (G_t)_{t \in I}$,
the shortest length of an exploration of $G$,
denoted by $\TEXP(G)$,
is the minimum length of a prefix $I'$ of $I$ 
such that there is an exploration of $(G_t)_{t \in I'}$, 
or $\infty$ if there is no exploration of $G$. 
If the temporal graph is always-connected, then any vertex can reach any other vertex within $n -1$ time steps, and therefore the length of a shortest exploration is always between $n-1$ and $n^2$. In this paper, we will only consider always connected temporal graphs,
and for this reason, we will for convenience assume they span $n^2$ time steps.

Deciding if a temporal exploration exists is NP-hard in general~\cite{MichailSpirakis16}. 
Even when the graph is always-connected, approximating 
$\TEXP(\cdot)$
up to a factor of $\mathcal{O}(n^{1-\varepsilon})$ is NP-hard for any $\varepsilon>0$~\cite{ErlebachHoffmannKammer21}. 
Moreover, the problem of deciding if an exploration of a given length exists remains NP-hard when the underlying graph has pathwidth $2$~\cite{pathwidthtwostillhard19}.

Due to the difficulty of approximating a shortest exploration in an always-connected temporal graph, there has been significant interest in determining upper and lower bounds on the length of such an exploration in terms of the number of vertices $n$.
There are temporal graphs that require $\Omega(n^2)$ time steps~\cite{ErlebachHoffmannKammer21ICALP2015version}, but the snapshots $G_t$ in these constructions are stars that contain vertices of degree $\Omega(n)$, and the underlying graph is very dense.
It is natural to ask whether restricting either the maximum degree of each snapshot $G_t$ or the density of the underlying graph can guarantee short explorations.
\begin{problem}
\label{prob:max_degree}
    For integers $1\leq\Delta\leq n$, what is the maximum value of $\TEXP(G)$ among all $n$-vertex temporal graphs $G=(G_1,\dots,G_{n^2})$ for which each snapshot $G_t$ is a connected graph of maximum degree at most~$\Delta$?
\end{problem}

The star construction can be extended to show that for every $\Delta$, it is possible that each snapshot $G_t$ has maximum degree at most $\Delta$, yet any shortest exploration takes at least $\Omega(\Delta n)$ time steps~\cite{ErlebachHoffmannKammer21ICALP2015version}.
There are also examples where the underlying graph is planar of maximum degree $3$ and all $G_t$ are paths, but a shortest exploration can still have length $\Omega(n\log n)$~\cite{AdamsonGusevMalyshevZamaraev22, ErlebachHoffmannKammer21ICALP2015version}. 
Inapproximability results persist even when the underlying graph has bounded maximum degree. For example, for all $\varepsilon,\delta > 0$, if the underlying graph has maximum degree at most $\Delta =  \Omega(n^{\delta})$, it is still NP-hard to approximate the length of a shortest exploration within a factor of $\mathcal{O}(\Delta^{1-\varepsilon})$~\cite{ErlebachHoffmannKammer21}.

Nonetheless, in the case where each snapshot has bounded maximum degree, Erlebach and Spooner \cite{ErlebachSpooner18} gave the first subquadratic upper bound for the shortest exploration.
They showed that $\TEXP(G) = \mathcal{O}(n^2 \log \Delta/\log n)$ if each snapshot has maximum degree at most $\Delta$. 
This was substantially improved by Erlebach, Kammer, Luo, Sajenko, and Spooner~\cite{ErlebachKammerLuoSajenkoSpooner19} who showed that $\TEXP(G) = \mathcal{O}_\Delta(n^{7/4})$\footnote{We write $f(\Delta,n) = \mathcal{O}_\Delta(g(n))$ for two functions $f(\cdot,\cdot)$ and $g(\cdot)$ if there exists a function $h$ such that $f(\Delta,n) = h(\Delta)\mathcal{O}(g(n))$.}, provided that either the maximum degree of each snapshot is at most $\Delta$ or the temporal exploration is allowed to make two moves per time step.

In terms of restrictions on the underlying graph, Erlebach, Hoffmann, and Kammer~\cite{ErlebachHoffmannKammer21ICALP2015version,ErlebachHoffmannKammer21} showed that temporal graphs $G$ can be explored in $\mathcal{O}(k^{1.5}\, n^{1.5} \log n)$ time steps if the underlying graph has treewidth~$k$ and in $\mathcal{O}(n^{1.8}\log n)$ time steps if the underlying graph is planar.
Adamson, Gusev, Malyshev, and Zamaraev~\cite{AdamsonGusevMalyshevZamaraev22} improved these results, showing that $\TEXP(G) = \mathcal{O}(k\, n^{1.5} \log n)$ if the underlying graph of $G$ has treewidth~$k$ and $\TEXP(G) = \mathcal{O}(n^{1.75} \log n)$ when the underlying graph is planar.
In addition, linear or near-linear bounds have been obtained when the underlying graph is a cactus~\cite{tempexpcacti14}, the $2\times n$-grid~\cite{ErlebachHoffmannKammer21}, a cycle with few chords~\cite{AdamsonGusevMalyshevZamaraev22,ErlebachHoffmannKammer21,taghian2020exploring}, or a graph with small orbit number~\cite{DogeasErlebachKammerMeintrupMoses24}.
Variations on the temporal exploration problem have also been studied, such as bounding the number of edges used during the exploration~\cite{specialcases25minedges} or finding a short temporal exploration when few edges are removed from the underlying graph in each snapshot~\cite{specialcases22} or when some edges are present for a longer time interval~\cite{specialcases13}.

As bounding the length of a shortest exploration in temporal graphs turned out to be quite difficult, 
Baguley, G\"{o}bel, Klodt, Skretas, Sylvester, and Zamaraev~\cite{BaguleyGobelKlodtSkretasSylvesterZamaraevRST25} also recently introduced the following random model. 
Given a probability distribution $\mu$ on a collection of trees on $n$ vertices, they sample a random temporal graph $G = (G_1, \dots, G_{n^2})$ by sampling a random tree $G_t \sim \mu$ independently at each time step. 
In this setting, they showed that $\TEXP(G) = \mathcal{O}(n^{3/2})$ with high probability. 
While this bound is optimal for general trees, the best lower bound known when sampling trees of maximum degree at most $\Delta$ is $\Omega(\sqrt{\Delta}n)$. The authors remark that their upper bound argument does not yield any improvement on the asymptotics of $n^{3/2}$ when the maximum degree is bounded. They conjecture that the value of $\TEXP$ should match the lower bound.

\begin{conjecture}[Conjecture 1.3 in \cite{BaguleyGobelKlodtSkretasSylvesterZamaraevRST25}]\label{conj:initial_random_conjecture}
    There is a constant $C>0$ such that the following holds.
    Let $1\leq \Delta \leq n$ be integers and $\mu $ be a probability distribution on the set of trees of maximum degree~$\Delta$ on vertex set $[n]$. If $G=(G_1,\dots,G_{n^2})$ where each $G_t\sim \mu$ is sampled independently, 
    then
    \[
    \mathbb{P}(\TEXP(G)\leq C\cdot \sqrt{\Delta}\cdot n )\to1 \text{ as }n\to \infty.
    \]
\end{conjecture}

\subsection{Our contribution}

We significantly improve the best upper bounds known on the length of a shortest temporal exploration in temporal graphs whose snapshots have bounded maximum degree, temporal graphs whose underlying graph has small edge density, and random temporal graphs that satisfy degree constraints. In fact, we prove these bounds by introducing a new unified approach for studying the temporal exploration problem simultaneously across all of these settings. This allows us to improve several of the upper bounds mentioned above for deterministic temporal graphs~\cite{AdamsonGusevMalyshevZamaraev22,ErlebachHoffmannKammer21ICALP2015version,ErlebachHoffmannKammer21,ErlebachKammerLuoSajenkoSpooner19,ErlebachSpooner18}. In addition, we prove a nearly tight upper bound for random temporal graphs whose snapshots are trees with bounded maximum degree, settling the conjecture of Baguley et al.~\cite{BaguleyGobelKlodtSkretasSylvesterZamaraevRST25} up to a $(\sqrt{\Delta}\log^3 n)$-factor.

Our unified approach is based on a parameter that generalises both (i) bounded maximum degree $\Delta$ in each snapshot and (ii) bounded average degree $d$ of the underlying graph. 
For any vertex $v \in V(G)$, define $\dmax(v) \coloneqq \max_{t \in I} d_{G_t}(v)$ as the maximum degree of $v$ over all possible time steps. We define the \emph{average temporal maximum degree} of $G$ as $D \coloneqq \sum_{v \in V(G)} \dmax(v) / n$. In words, the average temporal maximum degree $D$ measures the average, over all vertices, of the largest number of neighbours a vertex can have at any single time step.

Because $D\leq d$ and $D\leq \Delta$, bounds on the length of a shortest exploration in terms of $D$ apply directly for both setting (i) and setting (ii). Note that results stated only for setting (i) or setting (ii) cannot be lifted to the other because $\Delta$ and $d$ are incomparable. 

\subsubsection*{Improved exploration of deterministic temporal graphs}

The following result captures all of our contributions on the length of a short exploration of deterministic temporal graphs in one single statement, using our new notion of average temporal maximum degree.

\begin{theorem}
    \label{thm:main}
    For any $n$-vertex always-connected temporal graph $G=(G_1, \dots G_{n^2})$ with average temporal maximum degree $D$,
    \[
        \TEXP(G) = \O(n^{3/2} \sqrt{D \log n}).
    \]
\end{theorem}

The average temporal maximum degree encompasses many classes of temporal graphs investigated previously.
In particular, for setting (i) where the maximum degree of each snapshot $G_t$ is at most $\Delta$, we obtain the following corollary.

\begin{corollary}
    For any $\Delta \in \NN$ and any $n$-vertex always-connected temporal graph $G=(G_1, \dots, G_{n^2})$ where $G_t$ has maximum degree at most $\Delta$ for all $t \in [n^2]$, 
    \[
        \TEXP(G) = \O(n^{3/2} \sqrt{\Delta \log n}).
    \]
\end{corollary}

This significantly improves the previous best bound of $\mathcal{O}_\Delta(n^{7/4})$ by Erlebach, Kammer, Luo, Sajenko, and Spooner~\cite{ErlebachKammerLuoSajenkoSpooner19}.
Moreover, in any always-connected temporal graph, the square of each snapshot $G_t$ has a spanning tree of maximum degree at most three \cite[Lemma 4.1]{ErlebachKammerLuoSajenkoSpooner19}, and so the corollary above also implies that $\mathcal{O}(n^{3/2} \sqrt{\log n})$ time steps suffice if the exploration is allowed to make two moves per time step.
For this problem, $\mathcal{O}(n^{7/4})$ was the previous best bound as well~\cite{ErlebachKammerLuoSajenkoSpooner19}.

Moreover, for setting (ii) where the underlying graph of $G$ has average degree at most $d$, we also obtain the following result.

\begin{corollary}
    For any $d \in \NN$ and any $n$-vertex always-connected temporal graph $G=(G_1, \dots, G_{n^2})$ whose underlying graph has average degree at most $d$,
    \[
        \TEXP(G) = O(n^{3/2} \sqrt{d \log n}).
    \]
\end{corollary}

No subquadratic upper bound on the length of a shortest exploration was previously known when the underlying graph has bounded average degree. Using this result as a black box, we also immediately obtain the following consequences for short explorations in temporal graphs whose underlying graph $H$ satisfies certain structural properties.
\begin{itemize}
    \item If $H$ is planar, then there exists an exploration spanning $\O(n^{3/2} \sqrt{\log n})$ time steps. The previous best bound was $\mathcal{O}(n^{7/4}\log n)$~\cite{AdamsonGusevMalyshevZamaraev22}.
    \item If $H$ has treewidth at most $k$, then there exists an exploration spanning $\O(n^{3/2} \sqrt{k \log n})$ time steps. This improves the previous best bound $\mathcal{O}(k\, n^{3/2}\log n)$ \cite{AdamsonGusevMalyshevZamaraev22} by a factor of order $\sqrt{k \log n}$.
    \item If $H$ is $K_t$-minor-free, then $H$ has average degree $\O(t \sqrt{\log t})$~\cite{kostochka1982minimum,Kostochka1984,Thomason1984}, and so there exists an exploration spanning $\O(t^{1/2} \log^{1/4} t \cdot n^{3/2} \sqrt{\log n})$ time steps. 
    \item If $H$ has at most $o(n^2/\log n)$ edges, then there exists an exploration spanning $o(n^2)$ time steps.
\end{itemize}

\paragraph{Algorithmic variants.}
The proof of Theorem \ref{thm:main} and its consequences can easily be made algorithmic. In fact, for all theorems stated above, there is a polynomial-time algorithm that, given a temporal graph, constructs an exploration satisfying the stated bounds.

\bigskip

\subsubsection*{Improved exploration of random temporal graphs}

Our second contribution provides strong evidence towards Conjecture~\ref{conj:initial_random_conjecture} by improving the previous best upper bound of $O(n^{3/2})$ to a bound whose dependence on $n$ is only quasilinear. In fact, our result applies again in the more general setting of average temporal maximum degree. Here, if $\mu$ is a distribution on graphs on the vertex set $[n]$, we define the average temporal maximum degree of $\mu$ as $D \coloneqq \sum_{v \in [n]} \dmax^\mu(v)/n$ where $\dmax^\mu(v) \coloneqq \max_{G: \mu(G) > 0} d_G(v)$ is the maximum degree of $v$ across all graphs that have non-zero probability under $\mu$. Observe that with this definition, if we sample a random temporal graph $G = (G_1, \dots, G_{n^2})$ with $G_1, \dots, G_{n^2} \sim \mu$ independently, then $G$ will have average temporal maximum degree at most $D$. Our main result for random temporal graphs is the following.

\begin{theorem}
    \label{thm:main_proba}
    Let $\mu$ be a probability distribution on connected graphs on the vertex set $[n]$ and let $D$ be the average temporal maximum degree of $\mu$. If $G=(G_1,\dots,G_{n^2})$ where $G_1,\dots,G_{n^2} \sim \mu$ independently, then
    \[
        \Ex[\TEXP(G)]=\mathcal{O}(D \cdot n\log^3n).
    \]
\end{theorem}

In particular, if $\mu$ is a distribution on trees with maximum degree at most $\Delta$, then $\Ex[\TEXP(G)]=\mathcal{O}(\Delta n\log^3n)$. In addition, as in the deterministic case, this result applies to a much broader family of distributions on graphs. For example, it recovers up to a $(\log^3 n)$-factor the result that there is an exploration in $\O(m)$ time steps if $\mu$ is a distribution on the subgraphs of a graph with $m$ edges \cite{ErlebachHoffmannKammer21,BaguleyGobelKlodtSkretasSylvesterZamaraevRST25}. Note that the above result is about the expected length of a short temporal exploration instead of being a statement that holds with high probability. However, from our results it is easy to deduce a statement of the same order that holds with high probability, as shown in Appendix~\ref{sec:equivalence}.

Baguley et al.~\cite{BaguleyGobelKlodtSkretasSylvesterZamaraevRST25} suggested that in this random model, a near-optimal shortest temporal exploration can be found via a  greedy approach based on the revealed temporal graph $(G_1,\dots,G_{n^2})$ by starting at an arbitrary vertex $v_1$ and, after reaching a vertex $v_k$ at time $t_k$, by choosing the next vertex $v_{k+1}$ to visit greedily as the (temporally) closest unvisited vertex when starting at $v_k$ in $(G_{t_k+1},\dots,G_{n^2})$. 

While our strategy also contains various `greedy' components, it differs significantly from the strategy above. Firstly, rather than searching for the next vertex to visit \textit{after} revealing the samples, we choose an order on the vertices \textit{before} sampling based only on the probability distribution $\mu$. Secondly, instead of always choosing one vertex at a time, we construct the vertex order in a more global fashion similar to the greedy algorithm for minimum weight spanning tree. Together, this allows us to obtain an order on the vertices of the graph such that the expected length of an exploration is short if we visit the unvisited vertices in this order. This result is stated below. Here, $\mathbb{E}_{\mu}[u\to v]$ denotes the expected length of a shortest temporal walk from $u$ to $v$ in the random temporal graph $G = (G_1, \dots, G_{n^2})$ where each snapshot $G_t$ is sampled independently from $\mu$. 

\begin{theorem}
    \label{thm:main_proba2}
    Let $\mu$ be a probability distribution on connected graphs on the vertex set $[n]$ and let $D$ be the average temporal maximum degree of $\mu$. Then there is an order $v_1,\dots,v_n$ on the vertex set $[n]$ that depends only on $\mu$ such that 
    \[
        \sum_{i=1}^{n-1} \mathbb{E}_{\mu}[v_i\to v_{i+1}]=\mathcal{O}(D \cdot n \log^3 n). 
    \]
\end{theorem}

Note that \cref{thm:main_proba} is a direct corollary of \cref{thm:main_proba2}. 

\paragraph{Related work: other random models}
Another natural random model to consider is the random temporal graph $G=(G_1,\dots,G_{n^2})$ where each $G_t\sim G(n,p)$ is sampled independently for some choice of $p=p(n)$. We provide a brief discussion of this setting  in Appendix~\ref{app:B}. 

The literature has already considered a slightly different random model called the \textit{random simple temporal graph} $\mathcal{F}_{n,p}$. Here, we sample a graph $G\sim G(n,p)$ together with a uniformly random total order $\sigma$ on the edges of $G$. A temporal path is then a path in $G$ whose edges appear in increasing order in $\sigma$. 
Even for $\mathcal{F}_{n,1}$, it was difficult to determine whether there is a temporal path that visits all vertices. This was eventually solved by Martinsson~\cite{martinsson} who showed that such paths exist with high probability, confirming a conjecture from Lavrov and Loh~\cite{LavrovLoh}. 
Several other temporal connectivity questions have been considered in this model; we refer the reader to~\cite{Angel_Ferber_Sudakov_Tassion_2020,Broutinea,Casteigtsea} for a further literature review.

\paragraph{Subsequent work.}
In this paper we focused on proving upper bounds on the length of a short exploration based only on degree constraints. We did not try to optimize the bounds for specific graph classes such as temporal graphs whose underlying graph is planar or has bounded treewidth. It would interesting to see what bounds can be obtained when combining our approach with the specific structure available in such classes beyond degree constraints. 
In this direction, after making this paper available as a preprint, Adamson and Spirakis~\cite{adamson2026exploration} used our methods to prove that any always-connected temporal graph whose underlying graph has treewidth $k$ can be explored in at most $\mathcal{O}(n^{4/3} k^{5.5} \log^{2.5} n)$ time steps.

\paragraph{Roadmap.} We first give a detailed sketch of our proofs in \cref{sec:sketch}. We then give a complete proof of \cref{thm:main} in \cref{sec:deterministic} and of \cref{thm:main_proba2} in \cref{sec:random}. 

\subsection{Proof sketch}
\label{sec:sketch}
A standard approach to constructing a temporal exploration of length $\O(n^2)$ in an always-connected temporal graph is to show that any vertex $u$ can reach any other vertex $v$ within $n-1$ time steps. We can prove this by considering the \emph{forward ball} $F(u,t)$ of $u$ at time $t$, which consists of all vertices that $u$ can reach within $t$ time steps. Since $G_t$ is connected, either $F(u,t-1)$ already contains all vertices of the graph, or there exists a vertex $w_{u,t} \in N_{G_t}(F(u,t-1))$.\footnote{We write $N_H(S) \coloneqq \{v \in V(H) \setminus S : v s \in E(H) \text{ for some } s \in S\}$ for the open neighbourhood of $S$ in a graph $H$.} In the latter case, $F(u,t)$ contains $w_{u,t}$ and so $F(u,t)$ is strictly larger than $F(u,t-1)$. In particular, $F(u,n-1)$ must contain all vertices, and so $u$ can reach $v$ within $n-1$ time steps, as claimed.

This argument also immediately implies that if there are $k$ vertices that have not yet been explored, then $u$ can reach at least one unexplored vertex within at most $n-k$ time steps. If, however, each snapshot is always the same connected graph, then a stronger bound would hold: there would be two unexplored vertices $u$ and $v$ connected by a temporal walk of length at most $\O(n/k)$.

The following lemma is key to both our deterministic and random results. It proves that such an $\O(n/k)$-bound still holds for temporal graphs if they are always-connected and have bounded average temporal maximum degree.

\begin{restatable}{lemma}{keylemma}
\label{lemma:no_big_stable_set}
    Let $G = (G_t)_{t \in I}$ be an always-connected temporal graph with average temporal maximum degree $D$ and $n$ vertices, and let $X \subseteq V(G)$ be a set of at least two vertices.
    If $\abs{I} \geq 2 D n / \abs{X} + 1$, then there exist two distinct vertices $u,v \in X$ such that there is a temporal walk from $u$ to $v$ in $G$.
\end{restatable}
We first sketch the proof of this key lemma and then sketch how we can derive Theorem~\ref{thm:main} and \cref{thm:main_proba2} from it.

\paragraph{Proof sketch for the key lemma.} To see why there is a temporal walk between two of the $k$ unexplored vertices in $\O(Dn/k)$ time steps, let us first argue why such a walk exists if every snapshot is the same graph. For this, recall the definition of $w_{u,t}$ from above, and say that $u$ \emph{records} $w_{u,t}$ at time step $t$. Observe that every vertex $u$ records any other vertex at most once. Moreover, after $n/k+1$ time steps, the total number of recorded vertices is greater than $n$. So, there must be two unexplored vertices $u$ and $v$ and time steps $t_1 \le t_2$ such that $w_{u,t_1} = w_{v,t_2} \eqqcolon w$. This means that there is a walk from $u$ to $w$ in $t_1$ time steps and a walk from $v$ to $w$ in $t_2$ time steps. As every snapshot is the same graph, we can combine these walks to obtain a walk from $u$ to $v$ in $t_1 + t_2 = \O(n/k)$ time steps, as previously claimed.

The crucial step in this argument is that if $w_{u,t_1} = w_{v,t_2}$ for $t_1 < t_2$, then there is a temporal walk from $u$ to $v$. If we still want this to hold even if the snapshots of the temporal graph are not all the same, then we not only need that $w_{u,t_1} \in F(u,t_1)$, but we also need that $w_{v,t_2} \in B(v,t_2)$ where $B(v,t)$ is the \emph{backward ball} of $v$ at time $t$, which consists of all vertices that can reach $v$ with a temporal walk that starts at time step $t$. Therefore, for every time step $t$, we want that every unexplored vertex $u$ records a vertex $w_{u,t} \in F(u,t) \cap B(u,t)$. Then, if $w_{u,t_1} = w_{v,t_2}$ for $t_1 < t_2$, we get the desired temporal walk from $u$ to $v$.

Note that after $2Dn/k + 1$ time steps, some vertex $w$ must have been recorded strictly more than $2 \cdot d_{\max}(w)$ times. We would like to show that this implies that there are two \emph{distinct} unexplored vertices $u$ and $v$ and time steps $t_1 < t_2$ with $w_{u,t_1} = w = w_{v,t_2}$. However, this could fail if $w$ is always recorded by the same vertex. Previously, any vertex $u$ recorded $w$ at most once because $w_{u,t} \in N_{G_t}(F(u,t-1))$. To replicate this, let
\[
    w_{u,t} \in N_{G_t}(F(u,t-1) \cap B(u,t+1)) \subseteq F(u,t) \cap B(u,t).
\]
With this choice, it can be shown that $u$ records any other vertex at most twice. Since $w$ was recorded strictly more than $2 \cdot d_{\max}(w)$ times, it then follows that $w$ was recorded by strictly more than $d_{\max}(w)$ distinct vertices.

Unfortunately, one last problem could be that all these vertices record $w$ at the same time step $t$, which would not immediately allow us to obtain the desired temporal walk as we cannot move across two edges at time step $t$. However, since there are strictly more than $d_{\max}(w) \ge d_{G_t}(w)$ of these vertices, this would imply that there must be a neighbour $x \in N_{G_t}(w)$ and two unexplored vertices $u$ and $v$ such that $x \in F(u,t-1) \cap B(u,t+1)$ and $x \in F(v,t-1) \cap B(v,t+1)$. But this directly gives a temporal walk from $u$ to $v$ via $x$, as required.

\paragraph{Proof sketch of Theorem~\ref{thm:main}} 
The proof of Theorem~\ref{thm:main} first uses \cref{lemma:no_big_stable_set} to show that within $n$ time steps, there exists a temporal walk that visits at least $\Omega(\sqrt{k/(D \log k)})$ unexplored vertices. We first sketch how \cref{thm:main} follows from this. Namely, we divide the timeline into intervals of $n$ consecutive time steps and use every second interval to explore at least $\Omega(\sqrt{k/(D \log k)})$ of the remaining $k$ vertices that are unexplored at the beginning of the interval. 
The $n$ time steps in-between allow us to move from the end of one partial exploration to the start of the next. 
A calculation then shows that this strategy explores all vertices within at most $\O(n^{3/2} \sqrt{D \log n})$ time steps, as required. 

We next sketch how to show that within $n$ time steps, there is a temporal walk that visits at least $\Omega(\sqrt{k/(D \log k)})$ of the $k$ unexplored vertices. To do this, we will  divide the $n$ time steps into smaller time intervals and ensure that we always visit at least one new unexplored vertex in each interval. If each interval $I$ consists of more than $4Dn/k$ time steps, then we could apply the key lemma to any set of $k/2$ unexplored vertices to obtain a temporal walk in $I$ between two distinct unexplored vertices from this set. This would show that at least $k/2$ unexplored vertices have a temporal walk in $I$ to a different unexplored vertex. Unfortunately, when trying to connect these temporal walks across multiple time intervals, nothing prevents the walk from returning to an unexplored vertex it had visited more than one time interval ago, and so we could not guarantee that we visit many \emph{distinct} unexplored vertices.

Instead, we will find a small set of unexplored vertices $S_1$ that can reach all other unexplored vertices within the first time interval. Then we simply ignore all vertices from $S_1$ in later time intervals to ensure that we only visit distinct unexplored vertices. Afterwards, in the second time interval, we again find a small set of unexplored vertices $S_2$ that can reach all remaining unexplored vertices, and we again ignore those vertices in the future. If we continue this process iteratively, we obtain disjoint sets of unexplored vertices $S_1, S_2, \dots$ where each vertex in $S_{i+1}$ can be reached by a temporal walk in the time interval $i$ that starts at some vertex in $S_i$. By working backwards, this allows us to construct a temporal walk that visits one new unexplored vertex in each time interval, as intended.

To find a small set of unexplored vertices that can reach all other unexplored vertices within a single time interval $I$, we repeatedly select an unexplored vertex $u$ that can reach many other unexplored vertices within $I$, remove $u$ as well as the vertices that $u$ can reach from the set of unexplored vertices, and then repeat this process until all unexplored vertices have been removed. Note that the set of vertices that we select during this process can reach all other unexplored vertices within the time interval $I$. Moreover, if $I$ is large enough, then, by our key lemma, for any small set of unexplored vertices there is a temporal walk between two distinct vertices from this set. This implies that there are temporal walks between many pairs of unexplored vertices, and so there exists a vertex that can reach many other unexplored vertices within $I$. In particular, for every vertex we select, we will remove many unexplored vertices, and so the set of selected vertices is small, as required.

If we keep track of the size of the small set of unexplored vertices that we obtain in each time interval, we can show that we can run this process for $\Omega(\sqrt{k/(D \log k)})$ intervals before we run out of vertices. As a result, we obtain a temporal walk that visits $\Omega(\sqrt{k/(D \log k)})$ distinct unexplored vertices within $n$ time steps.

\paragraph{Proof sketch of \cref{thm:main_proba2}.}
To prove \cref{thm:main_proba2}, our central claim is that for every subset of vertices $S\subseteq[n]$ with $|S|\geq 2$, there exist two distinct vertices $u,v\in S$ such that the expected length of a shortest temporal walk from $u$ to $v$ satisfies
\[
    \Ex_\mu[u\rightarrow v] = \O\left(\frac{D \cdot n\log^2 n}{|S|}\right).
\]
Note that the pair $(u,v)$ is chosen using only the distribution $\mu$ and does not depend on the particular temporal graph that is eventually sampled. 

From this, we can build the desired sequence of vertices $v_1,\ldots,v_n$ that describes an order in which the vertices of the graph can be explored with a short expected exploration length. We start with $n$ sequences of length $1$ each containing a distinct vertex of the graph. We then repeatedly concatenate sequences from this collection until we obtain a single sequence of length $n$. In each step, we apply the upper bound mentioned above to the set $S$ containing one endpoint from each of the sequences. This gives a pair of endpoints that can be joined by a short temporal walk in expectation. By concatenating the two corresponding sequences, we can reduce the total number of sequences by one while increasing the total expected exploration length by at most $\O(D n \log^2 n / s)$ where $s \coloneqq |S|$ is the current number of sequences. Once only one sequence $v_1,\ldots,v_n$ remains, summing up these increases in the total expected exploration length implies that 
\[
    \Ex[\TEXP(G)] \leq \sum_{i=1}^{n-1} \Ex_\mu[v_i\rightarrow v_{i+1}] \le \O\left(\sum_{s=2}^n \frac{D \cdot n \log^2 n}{s}\right) = \O(D \cdot n\log^3 n).
\]

It remains to explain the claim made above that there always exists a pair of distinct vertices $u,v \in S$ such that $\Ex_\mu[u \rightarrow v]=\mathcal{O}(D \cdot n\log ^{2}n/|S|)$. We would like to stress that even though Lemma~\ref{lemma:no_big_stable_set} shows that after revealing $\Theta(D n/|S|)$ snapshots there is always a pair of vertices connected by a temporal walk, this does \textit{not} imply that there exist two fixed vertices $u,v\in S$ for which the expected length of a temporal walk from $u$ to $v$ is $\O(D n/|S|)$. In fact, for the uniform distribution $\mu$ on paths on vertex set $[n]$ with $S=[n]$, we have $\Ex_\mu[u \rightarrow v] = \Omega(\log n)$ for every pair of distinct vertices $u,v \in S$, and so at least one additional $(\log n)$-factor is actually necessary in the claim.

We now sketch why two additional $(\log n)$-factors suffice. To prove this, we split a time interval of length $\mathcal{O}(D \cdot n\log ^{2}n/|S|)$ into blocks of length $L = \Theta(D n/|S|)$. \cref{lemma:no_big_stable_set} implies that within such a block $(G_k,\ldots,G_{k+L})$, every reasonably large subset of $S$ contains two distinct vertices $u$ and $v$ such that $u$ can reach $v$ in $(G_k,\ldots,G_{k+L})$. Using this repeatedly on large enough subsets of $S$, one can greedily construct a large directed matching in $S$, where every arc $(u,v)$ of the matching corresponds to a temporal walk from $u$ to $v$ inside the block.

In addition, all blocks are sampled independently from the same distribution. This implies that we can reduce the original problem to the following one: at each step we sample a large directed matching on $S$ independently from the same distribution, a temporal walk can traverse one arc at each step, and we need to prove that in this matching process, some fixed pair of vertices $u,v\in S$ has expected temporal distance $\O(\log^2 n)$. Since each arc corresponds to a temporal walk of length $L$ in the original temporal graph, this would imply the claim.

To prove this matching statement, we use an expansion argument. Consider the weighted complete graph on $S$ where the weight of a pair $uv$ is the probability that $uv$ appears in the sampled matching. Since the matchings are large, the total weight across all pairs of vertices is also large. By a standard minimality argument, this allows us to find a subset $A \subseteq S$ such that every subset $B \subseteq A$ must have noticeable weight going from $B$ to its complement $A \setminus B$. Note that this means that when we sample a single matching, the expected number of edges in the matching that go between $B$ and $A \setminus B$ is quite large.

This expansion property allows us to show that any vertex $u \in A$ can likely reach at least half of the vertices of $A$ within at most $\O(\log^2 n)$ sampled matchings. To prove this, we will keep track of how many vertices $u$ has reached after $t$ sampled matchings. Then, by the expansion property, the next matching will likely have many edges between this set and the remaining vertices of $A$, and so $u$ can likely reach an even larger number of vertices after $t+1$ sampled matchings. Finally, after $\O(\log^2 n)$ matchings, $u$ will likely have reached at least half of the vertices of $A$. A standard argument allows us to convert this result to a fixed vertex $v \in A$ such that the expected temporal distance from $u$ to $v$ is at most $\O(\log^2 n)$.

\section{Preliminaries}

We write $[\ell,r] \coloneqq \{\ell,\ldots,r\}$ and $[r] \coloneqq [1,r]$. All logarithms in this paper have base $2$. For a graph $H = (V,E)$, we write $N_H(S)\coloneqq \{v\in V\setminus S:vs\in E\text{ for some }s\in S\}$ for the open neighbourhood of a set of vertices $S \subseteq V$ and $d_{H}(u)$ for the degree of a vertex $u \in V$ in $H$. The \defin{average degree} of $H$ is defined as $\overline{d}(H) \coloneqq \frac1n \sum_{v\in V} d_H(v)$. A \defin{matching} is a set of vertex-disjoint edges.

A \defin{temporal graph} $G$ with vertex set $V$ over an interval $I \subseteq \mathbb{N}$ is a sequence of graphs $(G_t)_{t \in I}$ with $V(G_t) = V$ for all $t \in I$. We call $G_t$ the \defin{snapshot} at time step $t$, and we say that $(V,\bigcup_{t\in I} E(G_t))$ is the \defin{underlying graph} of $G$. If $G_t$ is connected for every $t \in I$, we say that $G$ is \defin{always-connected}.

A \defin{temporal walk} in $G = (G_t)_{t \in I}$ is a sequence of vertices $(v_t)_{t \in I \cup \{\max I + 1\}}$ such that for all $t \in I$ either $v_t v_{t+1} \in E(G_t)$ or $v_t = v_{t+1}$. A \defin{exploration} is a temporal walk starting at a prescribed vertex and visiting all vertices of $G$ at least once. The shortest length $\TEXP(G)$ of an exploration of $G$ is formally defined as the minimum length of a prefix $I' \subseteq I$ such that there is an exploration of $(G_t)_{t \in I'}$ from any starting vertex $v \in V(G)$, or $\infty$ if there is no such prefix.

In the rest of the paper we will work with the average temporal maximum degree of a temporal graph $G = (G_t)_{t \in I}$, which generalises both bounded degree in each snapshot and bounded average degree in the underlying graph. Recall that if $d_{\max}(v) \coloneqq \max_{t \in I} d_{G_t}(v)$ for any $v \in V(G)$, we defined the \defin{average temporal maximum degree} of $G$ as
\[
    D(G) := \frac1n \sum_{v \in V(G)} d_{\max}(v).
\]
In particular, if $G_t$ has maximum degree $\Delta$ for all $t \in I$, then $d_{\max}(v) \leq \Delta$ for all $v \in V(G)$, and thus $D(G) \leq \Delta$. Similarly, if $H$ is the underlying graph of $G$ and $\overline{d}(H) \leq d$, then since $d_{\max}(v) \leq d_H(v)$, we also have $\D(G) \leq d$.

We also study random temporal graphs. Given a distribution $\mu$ over all graphs on a fixed vertex set $V$, the \defin{random temporal graph} with distribution $\mu$ is defined as $G = (G_t)_{t \in I}$ where each snapshot $G_t$ is sampled independently from $\mu$. We sometimes write $(G_t)_{t \in I} \sim \mu$ to denote this graph. The average temporal maximum degree of $\mu$ is $D(\mu) := \frac1n \sum_{v \in V} d_{\max}^{\mu}(v)$ where $d_{\max}^{\mu}(v) := \max_{G:\mu(G)> 0} d_G(v)$. Note that this definition implies that $G = (G_t)_{t \in I} \sim \mu$ always satisfies $D(G) \leq D(\mu)$.

\section{Proof of Theorem~\ref{thm:main}}
\label{sec:deterministic}

Observe that for any $k$ vertices in a connected graph on $n$ vertices, there is always a path between two of these vertices of length at most $2n/k$. The key step in our argument is the following lemma which proves an analogous result for temporal graphs, quantifying the number of time steps required to guarantee that there is a temporal walk between two distinct vertices in a given subset of vertices.

\keylemma*

We remark that this lemma is tight up to the constant factor for $\abs{X} \ge D$ (see \cref{lemma:n-1_steps_to_go_anywhere} for $\abs{X} < D$). To illustrate that, consider the temporal graph $G = (G_t)_{t \in I}$ where $I \coloneqq [\ell]$ with $\ell \ge D$ and for $t\in [\ell]$, the graph $G_t$ consists of an $\ell \times m$ grid with $D - 4$ leaves attached to each vertex in row $t$, where each snapshot $G_t$ uses the same set $X$ of leaves (see \cref{fig:gridconstruction}).
Then, it can be shown inductively that for any $t \in [\ell]$, a temporal walk in $G_{[t]}$ that starts at some vertex $x \in X$ can only remain at $x$ or must end at a grid vertex in one of the first $t$ rows. 
This implies that there exists no temporal walk in $G$ between two distinct vertices of $X$. However,  $n = \Theta((\ell + D) m)$ and $\abs{X} = \Theta(D m)$, and therefore $D n / \abs{X} = \Theta(\ell + D) = \Theta(\abs{I})$. 

\begin{figure}[ht]
    \centering
    \includegraphics{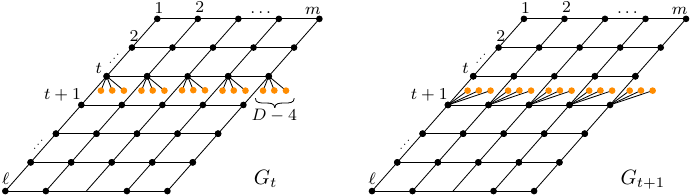}
    \caption{The figure depicts two consecutive snapshots from the construction that shows that \cref{lemma:no_big_stable_set} is tight. The orange vertices represent the set $X$.}
    \label{fig:gridconstruction}
\end{figure}

\begin{proof}[Proof of Lemma~\ref{lemma:no_big_stable_set}]
    Let $(G_t)_{t\in I}$ be an always-connected temporal graph with $n$ vertices and average temporal maximum degree $D$. Let $I = [\ell,r]$ and let $X \subseteq V(G)$ be a set of at least two vertices with $\abs{I} \ge 2Dn/\abs{X}+1$. 
    
    For all $u \in X$ and $t \in I$, we define the \emph{forward ball} and \emph{backward ball} of $u$ at time $t$ as follows:
    \begin{align*}
        F(u,t) & \coloneqq \{ w \in V : \text{there is a temporal walk from } u \text{ to } w \text{ in } G_{[\ell,t]}\} \text{ and} \\
        B(u,t) & \coloneqq \{ w \in V : \text{there is a temporal walk from } w \text{ to } u \text{ in } G_{[t,r]}\}.
    \end{align*}
    We also define $F(u,\ell-1) = B(u,r+1) \coloneqq \{u\}$. Note that $u\in F(u,\ell-1)\subseteq F(u,\ell)\subseteq \dots \subseteq F(u,r)$ and $u \in B(u,r+1) \subseteq B(u,r) \subseteq \dots \subseteq B(u,\ell)$. Moreover, if there exist distinct $u, v \in X$ and time steps $t_1 < t_2$ such that $F(u,t_1) \cap B(v,t_2) \neq \emptyset$, then there is a temporal walk from $u$ to $v$ in $G$. Indeed, if $x \in F(u,t_1) \cap B(v,t_2)$, then there is a temporal walk from $u$ to $x$ in $G_{[\ell,t_1]}$ and a temporal walk from $x$ to $v$ in $G_{[t_2,r]}$, and so by waiting at $x$ for $t_2 - t_1 - 1$ time steps, the concatenation of these walks yields a temporal walk from $u$ to $v$.

    We claim that for all $u \in X$ and $t \in I$, there exists a vertex
    \[
        w_{u,t} \in N_{G_t}(F(u,t-1) \cap B(u,t+1)).
    \]
    Indeed, let $S \coloneqq F(u,t-1) \cap B(u,t+1)$. Note that $S \neq \emptyset$ since $u$ is contained in all of its forward and backward balls. We may also assume that $S \neq V(G)$ since otherwise $u$ could reach all vertices in $X$ by a temporal walk in $G_{[\ell,t-1]}$ and so we would be done. Since $S$ is neither the empty set nor $V(G)$, the fact that $G_t$ is connected implies that there exists an edge between $S$ and $V(G) \setminus S$ in $G_t$. In particular, there exists a vertex $w_{u,t} \in N_{G_t}(S)$, as claimed. Fix such a vertex and say that $u$ \emph{records} $w_{u,t}$ at time step $t$.

    Since $N_{G_t}(S) \cap S = \emptyset$, we have
    \begin{equation}\label{eq:wutnotinintersection}
        w_{u,t} \notin F(u,t-1) \cap B(u,t+1).
    \end{equation}
    Moreover, we know that $F(u,t-1) \cup N_{G_t}(F(u,t-1)) = F(u,t)$ and $B(u,t+1) \cup N_{G_t}(B(u,t+1)) = B(u,t)$. This implies that $N_{G_t}(F(u,t-1) \cap B(u,t+1)) \subseteq F(u,t) \cap B(u,t)$, and so we get that
    \begin{equation}\label{eq:wutinintersection}
        w_{u,t} \in F(u,t) \cap B(u,t).
    \end{equation}
    See also \cref{fig:independent_set_figure} for an illustration of this situation. Note that if there exist two distinct vertices $u, v \in X$ and two time steps $t_1 < t_2$ such that $w_{u,t_1} = w_{v,t_2} \eqqcolon w$, then by \cref{eq:wutinintersection} we have $w \in F(u,t_1) \cap B(v,t_2)$, and so as argued before it follows that there is a temporal walk from $u$ to $v$ in $G$, as desired. 

    \begin{figure}
        \centering
        \includegraphics[page=2,scale=0.8]{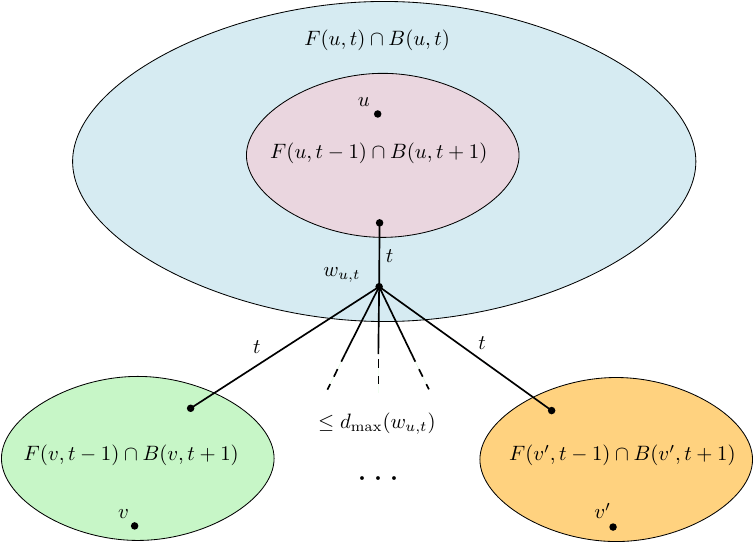}
        \caption{The figure depicts the situation described in the proof of \cref{lemma:no_big_stable_set}. The fact that $G_t$ is connected implies the existence of a vertex $w_{u,t}$ outside of but adjacent to $F(u,t-1) \cap B(u,t+1)$. By construction, $w_{u,t}$ will be in $F(u,t)\cap B(u,t)$. Crucially, $d_{G_t}(w_{u,t}) \leq d_{\max}(w_{u,t})$. Therefore, $w_{u,t}$ is \emph{recorded} by at most $d_{\max}(w_{u,t})$ vertices at time step $t$. We stress that even though some sets are represented to be disjoint for clarity, this is not always the case.}
        \label{fig:independent_set_figure}
    \end{figure}
    
    We claim that for any $w \in V(G)$ and $u \in X$, there are at most two time steps $t \in I$ such that $w_{u,t} = w$. Indeed, if there were three such time steps $t_1 < t_2 < t_3$, then as $w$ is recorded at time steps $t_1$ and $t_3$, \cref{eq:wutinintersection} implies that
    \[
        w \in F(u,t_1) \cap B(u,t_3) \subseteq F(u,t_2-1) \cap B(u,t_2+1).
    \]
    However, since $w = w_{u,t_2}$, we should also have $w \notin F(u,t_2-1) \cap B(u,t_2+1)$ by \cref{eq:wutnotinintersection}, a contradiction.

    Since every vertex in $X$ records one vertex at each of $\abs{I}$ different time steps, the total number of records is $\abs{X} \cdot \abs{I} > 2 D n = 2 \sum_{w \in V(G)} \dmax(w)$. In particular, there exists a vertex $w \in V(G)$ that is recorded strictly more than $2 \cdot \dmax(w)$ times. This implies that $w$ must have been recorded by strictly more than $\dmax(w)$ distinct vertices from $X$ since each vertex from $X$ records it at most twice.

    Recall that if there exist two distinct vertices $u, v \in X$ and two time steps $t_1 < t_2$ such that $w = w_{u,t_1} = w_{v,t_2}$, then there is a temporal walk from $u$ to $v$ in $G$, as desired. Therefore, we may assume that all vertices record $w$ at the same time step $t \in I$. When $w$ is recorded by a vertex $u \in X$ at time step $t$, this is because $w \in N_{G_t}(F(u,t-1) \cap B(u,t+1))$, and so there exists a vertex $x_u \in N_{G_t}(w)$ with $x_u\in F(u,t-1)\cap B(u,t+1)$. By definition, we know that $\dmax(w) \ge d_{G_t}(w)$. Since strictly more than $\dmax(w)$ distinct vertices from $X$ record $w$ at time step $t$, there must be two distinct vertices $u, v \in X$ with $x_u = x_v \eqqcolon x$. In particular, 
    \[
        x \in (F(u,t-1) \cap B(u,t+1))\cap (F(v,t-1) \cap B(v,t+1))\subseteq F(u,t-1)\cap B(v,t+1).
    \]
    Since $t - 1 < t + 1$, it follows, as argued before, that there is a temporal walk from $u$ to $v$ in $G$, as desired.
\end{proof}

When constructing a temporal exploration, we can apply the lemma above to the set $X$ of unexplored vertices, but also to any subset of $X$ of a fixed size. This provides us with several walks between unexplored vertices. Then, we could try to construct a temporal exploration by repeating this on multiple disjoint time intervals and linking the obtained walks across different time intervals. Unfortunately, it quickly becomes difficult to link these walks in such a way that they contain many \emph{distinct} unexplored vertices.

Instead, we use the preceding lemma to find a small subset of $X$ such that any other vertex in $X$ can be reached by a temporal walk starting from some vertex in this small subset. In later time intervals, we can then simply avoid this small subset entirely to ensure that we only visit distinct unexplored vertices.

\begin{lemma}\label{lemma:dominating_set}
    Let $G = (G_t)_{t \in I}$ be an always-connected temporal graph with average temporal maximum degree $D$, let $X \subseteq V(G)$ be a set of at least two vertices, and let $k \ge 2$ be an integer.
    If $\abs{I} \ge 2 D n / k + 1$, then there exists a set of vertices $S \subseteq X$ with
    \[
        \abs{S} \le \frac{\log \abs{X}}{-\log(1 - 1/ (2k))} \le 2k \log \abs{X}
    \]
    such that for all $u \in X$ there is a vertex $v \in S$ with a temporal walk from $v$ to $u$ in $G$.
\end{lemma}

\begin{proof}
    We prove the statement by induction on $\abs{X}$.
    If $\abs{X} \le 2k$, then the statement holds for $S \coloneqq X$.
    Otherwise, for any $v \in X$ define
    \begin{align*}
        F(v) & \coloneqq \{ u \in X \setminus \{v\} : \text{there is a temporal walk from } v \text{ to } u \text{ in } G\} \text{ and} \\
        B(v) & \coloneqq \{ u \in X \setminus \{v\} : \text{there is a temporal walk from } u \text{ to } v \text{ in } G\}.
    \end{align*}
    We claim that there is a vertex $v \in X$ such that $\abs{F(v) \cap X} \ge \abs{X}/2k - 1$. To show this, let $X_0 \coloneqq X$. Now suppose that $X_{i-1} \neq \emptyset$ for some $i \ge 1$. Because $\sum_{v \in X_{i-1}} \abs{F(v) \cap X_{i-1}} = \sum_{v \in X_{i-1}} \abs{B(v) \cap X_{i-1}}$, we can then select a vertex $v_i \in X_{i-1}$ such that $\abs{F(v_i) \cap X_{i-1}} \ge \abs{B(v_i) \cap X_{i-1}}$ and define $X_i \coloneqq X_{i-1} \setminus (F(v_i) \cup B(v_i) \cup \{v_i\})$. Let $\ell$ be minimal such that $X_\ell = \emptyset$. Note that for any $i \neq j$, there is no temporal walk from $v_i$ to $v_j$ in $G$, and so by \cref{lemma:no_big_stable_set} we get $\ell < k$. This implies that there exists some $i$ such that $\abs{(F(v_i) \cup B(v_i) \cup \{v_i\}) \cap X_{i-1}} \ge \abs{X} / (k-1)$ and so
    \[
        \abs{F(v_i) \cap X} \ge \abs{F(v_i) \cap X_{i-1}} \ge \frac{\abs{(F(v_i) \cup B(v_i)) \cap X_{i-1}}}{2} \ge \frac{\abs{X}}{2 (k-1)} - \frac{1}{2} \ge \frac{\abs{X}}{2k} - 1.
    \]
    This proves the claim.

    In particular, the set $X' \coloneqq X \setminus (F(v_i) \cup \{v_i\})$ has size at most $(1 - 1/(2k)) \cdot \abs{X}$.
    If $|X'| \leq 1$, then $S = \{v_i\} \cup X'$ has size at most $2 \leq 2k \log |X|$ and is as wanted.
    Otherwise, by applying the induction hypothesis to $X'$, we obtain a set of vertices $S' \subseteq X'$ with
    \[
        \abs{S'} \le \frac{\log \abs{X'}}{-\log(1 - 1/ (2k))} \le \frac{\log((1-1/(2k)) \cdot \abs{X})}{-\log(1 - 1/ (2k))} = \frac{\log \abs{X}}{-\log(1 - 1/ (2k))} - 1
    \]
    such that for all $u \in X'$ there is a vertex $w \in S'$ with a temporal walk from $w$ to $u$ in $G$.
    Therefore, $S \coloneqq S' \cup \{v_i\}$ satisfies the statement of the lemma.
\end{proof}

We now construct a temporal walk that covers many vertices of $X$ in $n$ time steps. We do this by dividing the $n$ time steps into multiple time intervals. In each time interval, we use the preceding lemma to find a small subset $S$ of $X$ that can reach all other vertices of $X$ within this time interval, and we remove all vertices of $S$ from $X$. We then show that these sets allow us to construct a temporal walk that visits exactly one vertex in each time interval from the small subset of vertices that we removed during that time interval. This way, all explored vertices will be distinct.

\begin{lemma}\phantomsection\label{lemma:exploring_X}
    Let $G = (G_t)_{t \in I}$ be an always-connected temporal graph with average temporal maximum degree $D$, and let $X \subseteq V(G)$ be a set of at least two vertices. If $\abs{I} \ge n$, then there exists a temporal walk in $G$ that covers at least $\frac18\sqrt{\abs{X}/(D \cdot \log \abs{X})}$ vertices of $X$.
\end{lemma}

\begin{figure}
    \centering
    \includegraphics[scale=1,page=3]{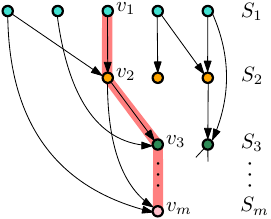}
    \caption{The figure depicts the situation described in the proof of \cref{lemma:exploring_X}. For each $i<j$ and $u\in S_j$, there is a vertex $v\in S_i$ such that there is a temporal walk from $v$ to $u$ in $G_{I_i}$. The red path depicts a potential temporal walk visiting $m$ vertices.}
    \label{fig:exploringX}
\end{figure}

\begin{proof}
    Let $m \coloneqq \big\lfloor \frac18\sqrt{\abs{X}/(D \cdot \log \abs{X})} \big\rfloor$ and $k \coloneqq \big\lceil\sqrt{D \cdot \abs{X} / \log \abs{X}}\big\rceil$.
    If $m = 0$, then any temporal walk that stays at a vertex in $X$ satisfies the statement, so we may assume that $m \ge 1$.
    We partition $I$ into $m$ intervals $I_1, \dots, I_m$ with at least $\floor{n/m}$ steps each.
    Iteratively for $i = 1, \dots, m$, we apply \cref{lemma:dominating_set} to obtain a set $S_i \subseteq X \setminus (S_1 \cup \dots \cup S_{i-1})$ with $\abs{S_i} \le 2 k \log \abs{X}$ such that for all $u \in X \setminus (S_1 \cup \dots \cup S_{i-1})$ there is a vertex $v \in S_i$ with a temporal walk from $v$ to $u$ in $G_{I_i}$.
    See Figure~\ref{fig:exploringX}.
    This is possible since for every $i \in [m]$,
    \[
        \abs{I_i} \ge \floor*{\frac{n}{m}} \ge \frac{n}{2 m} \ge 4 n \sqrt{\frac{D \cdot \log \abs{X}}{\abs{X}}} = 4 D n \sqrt{\frac{\log \abs{X}}{D \cdot \abs{X}}} \ge 4 D n / k \ge 2 D n / k + 1.
    \]
    The union of the sets $S_1, \dots, S_m$ has size at most
    \[
        m \cdot 2 k \log \abs{X} \le \frac{1}{8} \sqrt{\frac{\abs{X}}{D \cdot \log \abs{X}}} \cdot 2 \cdot 2 \sqrt{\frac{D \cdot \abs{X}}{\log \abs{X}}} \cdot \log \abs{X} = \frac{\abs{X}}{2} < \abs{X}.
    \]
    Therefore, there exists a vertex $v_{m+1} \in X \setminus (S_1 \cup \dots \cup S_m)$.
    Iteratively for $i = m, \dots, 1$, since $S_i$ satisfies \cref{lemma:dominating_set} there is a vertex $v_i \in S_i$ with a temporal walk from $v_i$ to $v_{i+1}$ in $G_{I_i}$. By concatenating these walks, this yields a temporal walk in $G$ that covers at least $m + 1$ vertices of $X$.
\end{proof}

Finally, we repeatedly apply the preceding lemma to find a temporal walk which covers many vertices that have not been covered so far. To ensure that we can reach the start of the next temporal walk from the end of the previous temporal walk, we always leave $n-1$ time steps between these walks as this allows any vertex to reach any other vertex during those time steps. This follows from the following lemma, first shown in \cite{ErlebachHoffmannKammer21}.

\begin{lemma}[\cite{ErlebachHoffmannKammer21}]\phantomsection\label{lemma:n-1_steps_to_go_anywhere}
    Let $G=(G_t)_{t \in I}$ be an always-connected temporal graph, 
    and let $u,v \in V(G)$ be two vertices.
    If $\abs{I} \geq n-1$, then there is a temporal walk from $u$ to $v$ in $G$.
\end{lemma}
Using this strategy, we now prove our main result for deterministic temporal graphs.

\begin{proof}[Proof of \Cref{thm:main}]
    Let $G = (G_t)_{t \in [n^2]}$ be a temporal graph.
    We iteratively construct a temporal walk $W$ of $G$ and a sequence of sets of vertices $V(G) \eqqcolon X_0 \supseteq X_1 \supseteq X_2 \supseteq \dots$ such that every vertex in $V(G) \setminus X_i$ has been covered by $W$ before time step $2 i n$. Suppose we have already constructed $X_i$ and the temporal walk $W$ up to step $2 i n$. 
    If $|X_i| \leq 2$, then by two consecutive applications of \Cref{lemma:n-1_steps_to_go_anywhere}, we can extend $W$ to visit all vertices in $X_i$ before time step $2(i+1)n$, and we set $X_{i+1} = \emptyset$.
    Otherwise, by \Cref{lemma:exploring_X}, there exists a temporal walk $W'$ in $G_{[(2i+1)n, 2(i+1)n-1]}$ that covers at least $\sqrt{\abs{X_i}/(D \cdot \log \abs{X_i})} / 8$ vertices of $X_i$.
    Observe that there exists a temporal walk between every pair of vertices in $G_{[2in+1, (2i+1)n-1]}$ by \Cref{lemma:n-1_steps_to_go_anywhere}. Therefore, we can extend $W$ by a temporal walk in $G_{[2in+1, (2i+1)n-1]}$ to the first vertex of $W'$, and then extend it by $W'$ itself. So, if $X_{i+1}$ denotes the set of vertices not covered before time step $2 (i+1) n$ by this extended temporal walk, then
    \begin{equation}\label{eq:improvement_on_Xi}
        \abs{X_{i+1}} \le \abs{X_i} - \frac{1}{8} \sqrt{\frac{\abs{X_i}}{D \cdot \log \abs{X_i}}}.
    \end{equation}
    
    We claim that there exists some $i_0 = \mathcal{O}(\sqrt{D \cdot n \log n})$ with $X_{i_0} = \emptyset$. Indeed, note that if 
    $\ell$ is an integer with $\ell \ge 4 \sqrt{2 D \cdot \abs{X_i} \log \abs{X_i}}$, then we can deduce from \eqref{eq:improvement_on_Xi} that
    \[
        \abs{X_{i+\ell}} \le \frac{\abs{X_i}}{2}.
    \]
    This follows from the fact that if $\abs{X_j} \ge \abs{X_i}/2$ for all $i \le j < i+\ell$, then \eqref{eq:improvement_on_Xi} implies that 
    $\abs{X_{j+1}} \le \abs{X_j} - \sqrt{\abs{X_i} / (2 D \cdot \log \abs{X_i})} / 8$ for all those $j$, and so $\abs{X_{i+\ell}} \le \abs{X_i} - \ell \cdot \sqrt{\abs{X_i} / (2 D \cdot \log \abs{X_i})} / 8 \le \abs{X_i}/2$. Therefore, if $\abs{X_i} \le n / 2^k$ for some $k$, then for $\ell_k \coloneqq \big\lceil 4 \sqrt{2 D \cdot (n / 2^k) \log(n / 2^k)} \big\rceil$ we get $\abs{X_{i+\ell_k}} \le n / 2^{k+1}$. This shows that $X_{i_0} = \emptyset$ for
    \[
        i_0 \coloneqq \sum_{k = 0}^{\lceil \log n \rceil} \ell_k \le \sqrt{D \cdot n \log n} \sum_{k = 0}^\infty \frac{8 \sqrt{2}}{2^{k/2}} = \mathcal{O}\left(\sqrt{D \cdot n \log n}\right).
    \]
    So, at time step $2 n i_0 = \mathcal{O}(n^{3/2} \sqrt{D \log n})$, all vertices have been covered.
\end{proof}

\section{Proof of Theorem~\ref{thm:main_proba2}}
\label{sec:random}

We now turn to the setting of random temporal graphs. Let $\mu$ be a probability distribution on a collection of connected graphs on the vertex set $[n]$, and let $D$ be the average temporal maximum degree of $\mu$. 

Observe that if $(G_1, \dots, G_k) \sim \mu$, then $(G_k, \dots, G_1) \sim \mu$, and that a temporal walk from $u$ to $v$ in $(G_1, \dots, G_k)$ yields a temporal walk from $v$ to $u$ in $(G_{k}, \dots, G_1)$. Hence, the probability that there is a temporal walk from $u$ to $v$ in $(G_1, \dots, G_k)$ equals the probability that there is a temporal walk from $v$ to $u$ in $(G_1, \dots, G_k)$. This yields the following observation.

\begin{observation}
    \label{lem:symmetry}
    For any $u,v\in [n]$, we have $\mathbb{E}_{\mu}[u\to v] = \mathbb{E}_{\mu}[v\to u]$.
\end{observation}

As a consequence, we may define the \textit{weight} of a pair of vertices $uv \coloneqq \{u,v\}$ as
\[
    w_{uv} \coloneqq \mathbb{E}_{\mu}[v\to u]=\mathbb{E}_{\mu}[u\to v].
\]
To prove Theorem~\ref{thm:main_proba2}, we need to find an order on the vertex set $[n]$ such that the sum of the weights of every pair of consecutive vertices is at most $\O(D \cdot n \log^3 n)$. The key to this is the following lemma which allows us to a find a single pair of vertices with small weight among any given set of vertices.

\begin{restatable}{lemma}{goodedge}\label{lemma:goodedge}
    There exists a constant $c > 0$ such that for any set $S\subseteq [n]$ of at least two vertices, there exist two distinct vertices $u,v\in S$ such that 
    \[
        \Ex_{\mu}[u\to v]\leq \frac{c D \cdot n \log^2 n}{|S|}.
    \]
\end{restatable}

Using this lemma, we can directly prove \cref{thm:main_proba2}. We construct the required order on all vertices by starting with a collection of smaller sequences of vertices that we merge one-by-one until we are left with just a single sequence that covers all vertices. To merge these sequences, we will apply the preceding lemma to the endpoints of the sequences to find a pair of endpoints with a small weight. By always merging the sequences at those endpoints, the final order will satisfy the upper bound on the weight that we need.

\begin{proof}[Proof of Theorem~\ref{thm:main_proba2}]
    We will define a sequence of distinct (unordered) pairs of vertices $e_1,\dots,e_{n-1}$ such that for each $i\in [n-1]$, the graph $F_i \coloneqq ([n],\{e_1,\dots,e_i\})$ that the pairs $e_1,\dots,e_i$ induce on $[n]$ is a linear forest (a disjoint union of paths) and
    \[
        w_{e_i} \leq \frac{cD \cdot n\log^2 n}{n-i+1},
    \]
    where $c$ is the constant from \cref{lemma:goodedge}.
    Let $F_0 \coloneqq ([n],\emptyset)$.
    Having defined the linear forest $F_i$ for some $i\in \{0,\dots, n-2\}$,
    we choose one endpoint per component of $F_i$. 
    Since $F_i$ is a forest with $i$ edges, it has $n-i$ components and so this defines a set $S$ of $n-i$ vertices.
    By Lemma~\ref{lemma:goodedge}, there is a pair of vertices $e_{i+1} = u v$ with $u,v\in S$ such that $w_{e_{i+1}}=\mathbb{E}_\mu[u\to v]\leq (c D \cdot n\log^2 n) / (n-i)$. 
    Note that $e_{i+1}\not \in E(F_i)$ and that $F_{i+1} \coloneqq ([n],\{e_1,\dots,e_{i+1}\})$ is a linear forest by construction.
    
    The final forest $F_{n-1}$ is a spanning path with weight
    \[
        \sum_{uv\in E(F_{n-1})} w_{uv} = \sum_{i=1}^{n-1} w_{e_i} \leq \sum_{i=1}^{n-1} \frac{c D \cdot n \log^2 n}{n-i+1} = c D \cdot n\log^2 n\sum_{i=2}^{n}\frac{1}{i} = \O(D \cdot n\log^3 n).
    \]
    This proves the theorem when $v_1, \dots, v_n$ is the ordering of vertices along the path $F_{n-1}$.
\end{proof}

It remains to prove \cref{lemma:goodedge}. From \cref{lemma:no_big_stable_set} we know that after $\O(D \cdot n / |S|)$ time steps, there is a temporal path between two distinct vertices of $S$. In fact, by repeatedly applying this lemma to subsets of $S$, we could actually show that there is a collection of $\Theta(|S|)$ disjoint pairs of vertices, each of which is connected by a temporal path within those time steps. Unfortunately, since the temporal graph is sampled randomly, this is a random collection of temporal paths, and so we do not immediately obtain a fixed pair of vertices from $S$ whose expected temporal distance is small.

However, since the random temporal graph has the same distribution in each interval of $\O(D \cdot n / |S|)$ time steps, the obtained collection of temporal paths has the same distribution in each of these intervals. We will show that by considering $\O(\log^2 |S|)$ consecutive intervals, this property allows us to find the desired fixed pair of vertices with small expected temporal distance.

To prove this, say that an oriented temporal matching is a temporal graph in which every snapshot is a matching consisting of oriented edges. Each snapshot will correspond exactly to the pairs of vertices that are connected by a temporal path in one interval. We say that a temporal matching $(\ora M_1,\dots,\ora{M}_k)$ has a  temporal walk from $u$ to $v$ if there is a sequence of vertices $w_1, \dots, w_{k+1}$ such that for every $t \in [k]$ either $w_t = w_{t+1}$ or $w_t w_{t+1} \in E(\ora M_t)$. This temporal walk in the temporal matching will then correspond to a temporal walk spanning multiple intervals in the original random temporal graph.

Since each interval in the random temporal graph is sampled independently according to the same distribution, we may assume that there is a fixed distribution $\ora\rho$ over directed matchings of $S$ and that each directed matching $\ora M_t$ is sampled independently according to $\ora\rho$. We say that the \emph{expected size of a matching} of $\ora\rho$ is $\Ex[|\ora M|]$ where $\ora M \sim \ora\rho$, and that $\ora\rho$ is \emph{symmetric} if for every pair of vertices $u, v \in S$ we have $\Pr((u, v) \in \ora M) = \Pr((v, u) \in \ora M)$ where $\ora M \sim \ora\rho$.

The following lemma shows that as long as $\ora\rho$ is symmetric and the expected size of a matching of $\ora\rho$ is a constant fraction of the size of $S$, it suffices to sample $\O(\log^2 |S|)$ directed matchings from $\ora\rho$ in order to guarantee that there is a fixed pair of distinct vertices $u, v \in S$ such that with probability at least $1/4$ there is a temporal walk from $u$ to $v$ in the sampled temporal matching.

\begin{lemma}\label{lemma:matchings}
    For every $0 < \delta \le 1/2$, there is a constant $C>0$ such that the following holds. Let $S$ be a set of at least two vertices and let $\ora\rho$ be a symmetric probability distribution on directed matchings on $S$ such that the expected size of a matching is at least $\delta|S|$. Then, if $k \coloneqq \lceil C\log^2 |S|\rceil$ and $(\ora M_1,\dots,\ora M_k)\sim\ora\rho$, there exist two distinct vertices $u,v \in S$ such that
    \[
        \Pr\big(\text{there is a temporal walk from $u$ to $v$ in $(\ora M_1, \dots, \ora M_k)$}\big) \geq \frac{1}{4}.
    \]
\end{lemma}

\begin{proof}
Let $\rho$ denote the probability distribution over \emph{undirected matchings} obtained from $\ora\rho$ by simply forgetting the directions. Consider a directed matching $\ora{M} \sim \ora{\rho}$ and the underlying undirected matching $M$, so by definition $M \sim \rho$. For any two distinct vertices $u,v\in S$, define the weight
\[
    \wt(\{u,v\}) \coloneqq \Pr(uv \in M).
\]
Since $\ora\rho$ is symmetric,
\begin{equation}
    \label{eq:undirected_to_directed}
    \Pr((u,v)\in \ora{M}) = \frac{\wt(\{u,v\})}{2}.
\end{equation}
Also recall that the expected size of $\ora{M} \sim \ora\rho$ and therefore $M$ is at least $\delta |S|$, and so
\[
    \sum_{\{u,v\}\in \binom{S}2} \wt(\{u,v\})\geq \delta |S|.
\]

For $A, B \subseteq S$, define
\[
    \wt(A, B) \coloneqq \sum_{u \in A, v \in B} \wt(\{u,v\}) \text{ and } \wt(A) \coloneqq \frac{\wt(A, A)}{2}.
\]
We choose $A\subseteq S$ as the smallest subset of $S$ with
\[
    \wt(A) > \frac{8\delta}{9} \frac{|A| \log |A|}{\log |S|}.
\]
Such a set $A$ always exists since the inequality above holds for $A \coloneqq S$. Note that $|A| \ge 2$ since any set $A$ with only one element satisfies $\wt(A) = 0$.

\begin{claim}
    \label{claim:B_outward_edges}
    For any $B\subseteq A$ with $1\leq |B|\leq |A|/2$, 
    \[
        \wt(B, A \setminus B)\geq \frac{4\delta |B|}{9\log |S|}.
    \]
\end{claim}
\begin{proofclaim}
Suppose towards a contradiction that the inequality above is false, then using the minimality of $A$ and the fact that $\abs{B} \le \abs{A}/2$ we get that
\begin{align*}
    \wt(A) & = \wt(B) + \wt(A \setminus B) + \wt(B, A \setminus B) \\
    & \le \frac{8\delta}{9} \frac{|B| \log |B|}{\log |S|} + \frac{8\delta}{9} \frac{|A \setminus B| \log |A \setminus B|}{\log |S|} + \frac{4\delta|B|}{9 \log |S|} \\
    & = \frac{8\delta}{9} \frac{|B| (\log |B| + 1/2) + |A \setminus B| \log |A \setminus B|}{\log |S|} \\
    & \le \frac{8\delta}{9} \frac{|B| \log |A| + |A \setminus B| \log |A|}{\log |S|} \\
    & = \frac{8\delta}{9} \frac{|A| \log |A|}{\log |S|}.
\end{align*}
This contradicts the definition of $A$.
\end{proofclaim}
Next, we consider the vertices of $A$ that are reachable from a subset of its vertices $B \subseteq A$ in one step of the random \emph{directed} matching. We denote this random set by
\[
    F(B, t) \coloneqq B \cup \{ v \in A : (u,v) \in \ora M_t \text{ for some } u \in B\}.
\]
Note that for a fixed subset $B \subseteq A$, the random set $F(B,t)$ has the same distribution for all $1\leq t\leq k$, and thus we will omit $t$ from the notation when irrelevant. By the previous claim, every small subset $B\subseteq A$ has good expansion with respect to the weight function $\wt$. We use this to prove the following claim which shows that $F(B)$ is usually significantly larger than $B$.

\begin{claim}
    \label{claim:B_expansion}
    Let $B \subseteq A$ with $1 \le |B| \le |A|/2$. Then, $\Ex\left[\frac{\abs{B}}{\abs{F(B)}}\right] \le 1 - \frac{\delta}{9 \log \abs{S}}$. 
\end{claim}
\begin{proofclaim}
    Let
    \[
        X \coloneqq \sum_{a \in A \setminus B} \mathds{1}\{(b,a) \in \ora{M_1} \text{ for some } b \in B\}.
    \]

    Note that $\abs{F(B)} = X + \abs{B}$. 
    Since $\ora{M_1}$ is a matching, the events $(b,a) \in \ora{M_1}$ and $(b',a) \in \ora{M_1}$ are disjoint for $b \neq b'$, and so
    \[
        \mathds{1}\{(b,a) \in \ora{M_1} \text{ for some } b \in B\} = \sum_{b\in B} \mathds{1}\{(b,a) \in \ora M_1\}.
    \]
    Recall from \eqref{eq:undirected_to_directed} that $\mathbb{P}((u,v)\in \ora{M_1})=\mathbb{P}(\{u,v\}\in M_1)/2 = \wt(\{a,b\})/2$.
    By the linearity of expectation and \cref{claim:B_outward_edges}, we can therefore deduce that
    \[
        \Ex[X] = \sum_{a \in A\setminus B} \sum_{b\in B} \Pr((b,a) \in \ora M_1) = \sum_{a \in A\setminus B} \sum_{b\in B} \frac{\wt(\{a,b\})}{2} = \frac{\wt(B,A\setminus B)}{2} \ge \frac{2\delta \abs{B}}{9\log \abs{S}}.
    \]
    Also, since $\ora{M_1}$ is a matching, we know that $0 \le X \le \abs{B}$, and so we have that
    \[
        \frac{\abs{B}}{\abs{F(B)}} = \frac{\abs{B}}{\abs{B} + X} = 1 - \frac{X}{\abs{B} + X} \le 1 - \frac{X}{2 \abs{B}}.
    \]
    In particular,
    \[
        \Ex\left[\frac{\abs{B}}{\abs{F(B)}}\right] \le 1 - \frac{\Ex[X]}{2 \abs{B}} \le 1 - \frac{\delta}{9 \log \abs{S}}. \qedhere
    \]
\end{proofclaim}

Fix an arbitrary vertex $u \in A$. We next show that with probability at least $1/2$, half of vertices of $A$ can be reached from $u$ in $\O(\log^2 |S|)$ steps. Indeed, let 
\[
    k \coloneqq \left\lceil\frac{18}{\delta}\log\abs{S}\log \abs{A} \right\rceil,
\]
and define $F_0 \coloneqq \{u\}$ and $F_t \coloneqq F(F_{t-1}, t)$ for $1\leq t \leq k$. Note that $F_t$ is the set of vertices that can be reached from $u$ in the first $t$ time steps of the temporal matching $(\ora M_1,\dots,\ora M_k)$.

\begin{claim}\label{claim:half-reached}
    We have $\Pr(|F_k| > |A|/2)\geq 1/2$.
\end{claim}
\begin{proofclaim}
    For $1\leq t\leq k$, let $X_t \coloneqq \abs A / \abs{F_t}$.
    Note that $X_{t+1}\leq X_t$.
    If $X_t\geq 2$, or equivalently $\abs{F_t}\leq |A|/2$, we have by Claim~\ref{claim:B_expansion} that 
    \begin{align*}
        \Ex[X_{t+1} \mid F_1,\dots,F_t] &= \Ex\left[\frac{X_{t+1}}{X_t} \ \Bigm\vert\ F_1,\dots,F_t\right] X_t \\
        &= \Ex\left[\frac{\abs{F_{t}}}{\abs{F_{t+1}}} \ \Bigm\vert\ F_1,\dots,F_t\right] X_t \leq \left(1-\frac{\delta}{9\log |S|}\right)X_t.
    \end{align*}
    Consequently, $\Ex[X_{t+1} \mid X_t\geq 2] \leq \left(1-\frac{\delta}{9\log |S|}\right)\Ex[X_t \mid X_t\geq 2]$.

    Suppose for a contradiction that $\Pr(\abs{F_{k}} > \abs{A}/2)< 1/2$. This means that $\Pr(X_{k} \ge 2) > 1/2$, and so $\Pr(X_t \ge 2) \ge \Pr(X_{k} \ge 2) > 1/2$ for all $t \le k$. Then
    \begin{align*}
        \Ex[X_{t+1}] & = \Pr(X_t \ge 2) \cdot \Ex[X_{t+1} \mid X_t \ge 2] + \Pr(X_t < 2) \cdot \Ex[X_{t+1} \mid X_t < 2] \\
        & \le \Pr(X_t \ge 2) \cdot \left(1 - \frac{\delta}{9 \log \abs{S}}\right) \cdot \Ex[X_t \mid X_t \ge 2] + \Pr(X_t < 2) \cdot \Ex[X_t \mid X_t < 2] \\
        & = \Ex[X_t] - \frac{\delta}{9 \log \abs{S}} \cdot \Pr(X_t \ge 2) \cdot \Ex[X_t \mid X_t \ge 2] \\
        & \le \Ex[X_t] - \frac{\delta}{18 \log \abs{S}} \cdot \Ex[X_t] \le e^{-\delta/(18 \log \abs{S})} \cdot \Ex[X_t],
    \end{align*}
    where in the penultimate step we used that $\Ex[X_t \mid X_t \geq 2] \geq \Ex[X_t]$ and $\Pr(X_t \geq 2) \geq 1/2$, and in the last step we used the fact that $1+x\leq e^x$ for all $x\in \mathbb{R}$.

    A simple induction then yields
    \[
        \Ex[X_{k}] \le e^{-k/(18 \log \abs{S}/\delta)} \cdot \Ex[X_0] = e^{-k/(18 \log \abs{S}/\delta)} \cdot \abs{A} < 1.
    \]
    This contradicts the fact that $X_t = \abs{A}/\abs{F_t} \ge 1$. Therefore, $\Pr(\abs{F_{k}} > \abs{A}/2) \ge 1/2$.
\end{proofclaim}

As a result, 
\begin{align*}
    \Ex[\abs{F_k\setminus\{u\}}] 
    &= \Ex[\abs{F_k}] - 1 \\
    &\geq \Pr(|F_k|\leq |A|/2) \cdot 1 
           + \Pr(|F_k|>|A|/2) \cdot (\lfloor |A| / 2\rfloor + 1) - 1 \\
    &\geq \frac{|A|-1}{4}.
\end{align*}
Rewriting the expectation, this means
\begin{align*}
    \frac{|A|-1}{4}
    \leq \Ex[\abs{F_k\setminus \{u\}}] 
    = \sum_{v\in A\setminus \{u\}} \Pr\big( \text{there is a temporal walk from $u$ to $v$ in $(\ora M_1, \dots, \ora M_k)$} \big).
\end{align*}
Therefore, there exists a vertex $v \in A \setminus \{u\}$ such that the probability that $(\ora M_1, \dots, \ora M_k)$ has a temporal path from $u$ to $v$ is at least $1/4$. This proves the lemma with $C = 18/\delta$. 
\end{proof}

Using the strategy outlined before, we are now equipped to prove Lemma~\ref{lemma:goodedge}, which we repeat below for the convenience of the reader. The proof of this lemma concludes the proof of \cref{thm:main_proba2}.
\goodedge*

\begin{proof}[Proof of Lemma~\ref{lemma:goodedge}]
Let $S\subseteq [n]$ with $|S| \geq 2$. We may assume that $n$ is sufficiently large and that $|S| = \Omega(\log ^2 n)$, since otherwise the statement follows immediately from the fact that $\Ex_\mu[u \rightarrow v] \leq n$ for all $u,v \in S$, which is a direct consequence of \cref{lemma:n-1_steps_to_go_anywhere}. In particular, $|S|/10 \geq 2$.

Let $\delta \coloneqq 9/20$, let $C$ be given by \cref{lemma:matchings}, and set
\[
    L \coloneqq \left\lceil \frac{2 D n}{\lfloor |S|/10\rfloor}+1 \right\rceil, \qquad
    k \coloneqq \lceil C \log^2 |S|\rceil, \qquad
\text{and} \qquad
    T \coloneqq kL.
\]

For any temporal graph $(G_1, \dots, G_L)$ on the vertex set $[n]$ with average temporal maximum degree at most $D$, construct a directed matching $\ora M = \ora M(G_1, \dots, G_L)$ on the vertex set $S$ as follows: starting with $\ora M \coloneqq \emptyset$, while there is a set of at least $|S|/10$ vertices uncovered by $\ora M$, add to $\ora M$ a pair $(u,v)$ of distinct uncovered vertices such that there is a temporal walk from $u$ to $v$ in $(G_1, \dots, G_L)$. Such a pair always exists by \cref{lemma:no_big_stable_set}. Moreover, we can define this process so that $\ora M(G_L, \dots, G_1)$ is obtained from $\ora M(G_1, \dots, G_L)$ by reversing the orientations of all the arcs. Then, when $(G_1, \dots, G_L) \sim \mu$, the directed matching $\ora M = \ora M(G_1, \dots, G_L)$ is a random variable over directed matchings on $S$ that cover at least $9|S|/10$ vertices, and for all distinct vertices $u,v \in S$ we have $\Pr((u, v) \in \ora M) = \Pr((v, u) \in \ora M)$.

Let $(G_1, \dots, G_T) \sim \mu$ and define $\ora M_t = \ora M(G_{(t-1)L+1}, \dots, G_{t L})$ for each $t \in [k]$. Then, by \cref{lemma:matchings} there exist two distinct vertices $u,v \in S$ such that
\[
    \Pr(\text{there is a temporal walk from } u \text{ to } v \text{ in } (\ora M_1,\dots, \ora M_k)) \geq \frac{1}{4}.
\]
Note that if there is a temporal walk from $u$ to $v$ in $(\ora M_1,\dots, \ora M_k)$, this implies that there is also a temporal walk from $u$ to $v$ in $(G_1,\dots, G_T)$, and therefore
\begin{equation}\label{eq:good_pair}
    \Pr(\text{there is a temporal walk from } u \text{ to } v \text{ in } (G_1,\dots, G_T)) \geq \frac{1}{4}.
\end{equation}

Consider now a random temporal graph $(G_1, \dots, G_n) \sim \mu$. For every $a,b\in [n]$ with $a \le b$, let $B(a,b)$ be the event that there is no temporal walk from $u$ to $v$ in $(G_a, \dots, G_b)$. By \eqref{eq:good_pair}, we know that $\Pr(B(a,a+T-1)) \leq 3/4$ for every $a \in [n-T+1]$. Observe that for disjoint intervals, these events are independent.

Let $X$ denote the random variable whose value is the minimum number of time steps $t \in [n]$ such that there is a temporal walk from $u$ to $v$ in $(G_1, \dots, G_t)$. By \cref{lemma:n-1_steps_to_go_anywhere}, $X$ is always at most $n-1$. We have
\[
    \Ex_\mu[u \to v] = \Ex[X] = \sum_{t=1}^{n} \Pr(X \geq t) = 1 + \sum_{t=2}^{n} \Pr(B(1,t-1)).
\]
Moreover, for every $t \in [n-1]$,
\[
    \Pr(B(1,t)) \leq \prod_{i=1}^{\lfloor t/T \rfloor} \Pr(B((i-1) T + 1,i T)) \leq \left(\frac{3}{4}\right)^{\lfloor t/T \rfloor}.
\]
Therefore,
\[
    \Ex_\mu[u \to v] \leq \sum_{t=1}^{n} \left(\frac{3}{4}\right)^{\lfloor (t-1)/T \rfloor} \leq T \cdot \sum_{i=0}^{\infty} \left(\frac{3}{4}\right)^i = 4T.
\]
Since $T = \O(D \cdot n \log^2 n / |S|)$, this proves the lemma.
\end{proof}

\section{Conclusion}\label{sec:conclusion}

Our first main contribution in this paper is that any always-connected temporal graph $G = (G_t)_{t \in [n^2]}$ on $n$ vertices admits a temporal exploration spanning at most $\mathcal{O}(n^{3/2}\sqrt{D\log n})$ time steps, where $D \coloneqq \sum_{u\in V(G)}\max_{t} d_{G_t}(u) / n$ denotes the average temporal maximum degree of $G$.
Although this significantly reduces the best upper bounds known for various natural classes of temporal graphs, such as when each snapshot $G_t$ has bounded maximum degree, when each $G_t$ is a path, or when the underlying graph is planar, the gap to the best lower bounds known remains large, where the best lower bound is $\Omega(n \log n)$ for each of these three examples.
It is unclear to us whether the lower or the upper bounds are closer to the truth. There is also a lower bound of $\Omega(\Delta n)$ if $\Delta$ is a bound on the maximum degree of each snapshot $G_t$, but this construction seems difficult to combine with the $\Omega(n\log n)$ lower bound. We therefore raise the following problem.

\begin{problem}
    Is there a function $f=\omega(1)$ such that for all integers $n\geq D\geq 1$, there is a temporal graph $G=(G_t)_{t \in [n^2]}$ on $n$ vertices with average temporal maximum degree at most $D$, yet the length of a shortest temporal exploration of $G$ is at least $f(D) \cdot  n \log n$?  \footnote{Update (October 12, 2026). A recent arXiv preprint from Carnevale~\cite{CarnevaleLowerbound} answers the question above positively by providing examples that take $\Omega(D n(1+\log(n/D))$ steps to explore.}
\end{problem}
One way to obtain such a lower bound construction would be via snapshots of maximum degree at most $\Delta$, but perhaps this more general setting allows for different approaches.

\medskip

Another direction for future research is to study whether our techniques can provide further improvements when additional structural constraints are placed on the temporal graph besides bounding its average temporal maximum degree.

\medskip

Our second main contribution is about the exploration of random temporal graphs where we obtain upper bounds that are quasi-linear in $n$ for the temporal exploration whenever each snapshot is sampled independently among connected graphs on $n$ vertices with constant maximum degree. In fact, we also prove this result in the more general setting of average temporal maximum degree. Compared to Conjecture~\ref{conj:initial_random_conjecture}, our upper bound has an additional $\sqrt{\Delta}\log^{O(1)}n$-factor. One of these additional $\log n$-factors is unavoidable in our greedy approach, as shown by the following proposition.

\begin{proposition}
    Let $\mu$ be the uniform distribution on paths on the vertex set $[n]$. For any order $v_1,\dots,v_n$ on the vertex set $[n]$,  
    \[
        \sum_{i=1}^{n-1} \mathbb{E}_{\mu}[v_i\to v_{i+1}]=\Omega(n \log n). 
    \]
\end{proposition}
\begin{proof}
    Fix any ordering $v_1,\dots,v_n$ on $[n]$ and let $G_1,\dots,G_{k}\sim \mu$ be independent. Since $\mu$ is a distribution on paths, for each $1\leq t\leq k$, each vertex has at most $2$ neighbours in $G_t$.
    In particular, the set $F(\{u\},k)$ of vertices that $u$ can reach in $(G_1,\dots,G_k)$ has size at most $2^{k+1}$, and so $\Ex[|F(\{u\},k)|]\leq 2^{k+1}$.
    
    Since $\mu$ is a uniform distribution on paths, for any two distinct vertices $v,w \in[n] \setminus \{u\}$ we have by symmetry that $\Pr(v\in F(\{u\},k)) = \Pr(w\in F(\{u\},k))$. From the linearity of expectation we thus get $\Pr(v\in F(\{u\}, k)) \leq 2^{k+1} / (n-1)$. Since $v\not\in F(\{u\}, k)$ precisely when is no temporal walk of length $k$ from $u$ to $v$, we obtain that
    \[
        \Ex_\mu[u\to v]\geq \sum_{k=1}^{\log n-2}\Pr(v\not \in F(\{u\},k))\geq \sum_{\ell=1}^{\log n-2}\frac12=\Omega(\log n). \qedhere
    \]    
\end{proof}

However, it is not clear whether the dependence on $\Delta$, or more generally $D$, can be improved in our greedy approach.

\begin{problem}
    \label{pb:order}
    Let $\mu$ be a probability distribution on connected graphs on the vertex set $[n]$ with maximum degree at most $\Delta$. Is there an order $v_1, \dots, v_n$ on the vertex set $[n]$ such that
    \[
        \sum_{i=1}^{n-1} \Ex_\mu[v_i\to v_{i+1}]=\sqrt{\Delta} \cdot n (\log n)^{\O(1)}?
    \]
\end{problem}

The constant $\O(1)$ above should be independent of $n$ and $\mu$. We note that the answer to \cref{pb:order} is `yes' for the example in~\cite{BaguleyGobelKlodtSkretasSylvesterZamaraevRST25} that proves the lower bound of order $\Omega(\sqrt{\Delta}n)$. Roughly speaking, this example consists of $n/\Delta$ many stars on $\Delta$ vertices, where the centre of each star is chosen randomly in each snapshot among a fixed subset of $\Delta/2$ vertices of the star. By the birthday paradox, the expected temporal distance between any two vertices within the same star is $\Theta(\sqrt{\Delta})$. This allows us to construct a vertex order $v_1,\dots,v_n$ for this example with $\sum_{i=1}^{n-1} \mathbb{E}_{\mu}[v_i\to v_{i+1}] = \Theta(\sqrt{\Delta} n)$.

\paragraph{Acknowledgements.} This paper was obtained by merging an earlier paper of PB, CG, LM and CR containing the results for the deterministic setting with the results in the random spanning tree model by PB, CG, LM and PS. CG and PS would like to thank the CWI PhaseCAP workshop in Combinatorics where the latter collaboration was started for a stimulating environment. We would also like to thank John Sylvester for helpful comments on the presentation of the paper, and Alp Müyesser for pointing us to McDiarmid's theorem \cite{mcdiarmid1980clutter}.

\paragraph{AI statement.} 
The results and writing in this paper are due to the authors. ChatGPT was used for proofreading purposes and literature search.

\bibliographystyle{bibstyle}
\bibliography{ref}

\newpage
\appendix

\section{Reducing to the expected length for a finite sequence}
\label{sec:equivalence}
We propose the following adjusted version of Conjecture~\ref{conj:initial_random_conjecture}.
\begin{conjecture}
    \label{conj:exp_version}
    There is a constant $C>0$ such that the following holds.
    Let $1\leq \Delta \leq n$ be integers and $\mu $ be a probability distribution on the set of trees of maximum degree~$\Delta$ on vertex set $[n]$. If $G=(G_1,\dots,G_{n^2})$ where $G_1,\dots,G_{n^2}\sim \mu$ independently, then
    \[
        \mathbb{E}[\TEXP(G)]\leq C \cdot \sqrt{\Delta} \cdot n.
    \]
\end{conjecture}
The same lower bound example as used in~\cite{BaguleyGobelKlodtSkretasSylvesterZamaraevRST25} for their conjecture shows that the conjecture above is best possible up to the constant $C$ (see also Section~\ref{sec:conclusion}). 

We already mentioned before that $(G_1,\dots,G_{n^2})$ admits a temporal exploration if each snapshot $G_t$ is a connected graph on vertex set $[n]$. This is why the original conjecture of~\cite{BaguleyGobelKlodtSkretasSylvesterZamaraevRST25}, which was stated for an infinite sequence of snapshots, is equivalent to the version we stated in Conjecture~\ref{conj:initial_random_conjecture}. 

The more relevant adjustment we propose in Conjecture~\ref{conj:exp_version} compared to the original conjecture is the switch from `high probability' to `expected length'. When ignoring $(\log n)$-factors, a short expected length implies a short length with high probability by the following observation.

\begin{proposition}
\label{prop:equivalence}
    Let $\mu$ be a probability distribution on connected graph on the vertex set $[n]$. If $G=(G_1,\dots,G_{n^2})$ where $G_1,\dots,G_{n^2}\sim \mu$ independently, then
    \[
        \Pr(\TEXP(G)\leq 2 \lceil\log n\rceil \cdot \Ex[\TEXP(G)]) \geq 1-1/n.
    \]
\end{proposition}
\begin{proof}
    By Markov's inequality, we have
    \[
        \Pr\big(\TEXP(G)\geq 2 \cdot \Ex[\TEXP(G)]\big)\leq \frac{1}{2},
    \]
    and so for $N \coloneqq 2 \cdot \Ex[\TEXP(G)]$,
    \[
        \Pr((G_1,\dots,G_N)\text{ has a temporal exploration})\geq \frac12.
    \]
    Using this fact on $\lceil\log n\rceil $ disjoint intervals of $N$ time steps each, we get that
    \[
        \Pr(\TEXP(G)\geq 2 \lceil\log n\rceil \cdot \Ex[\TEXP(G)])\leq \left(\frac{1}{2}\right)^{\log n} = \frac{1}{n}. \qedhere
    \]
\end{proof}
Conversely, if $\Pr\big(\TEXP(G) \leq T\big) \geq 1/2$ then $\Ex[\TEXP(G)]\leq 2T$ which can be shown in a similar fashion to the end of the proof of Lemma~\ref{lemma:goodedge}. Hence Conjecture~\ref{conj:initial_random_conjecture} implies Conjecture~\ref{conj:exp_version}.

Because of the way we prove Theorem~\ref{thm:main_proba}, in our case no additional $(\log n)$-factor is needed.

\begin{corollary}
    There is a constant $C>0$ such that the following holds. Let $\mu$ be a probability distribution on connected graph on the vertex set $[n]$ and let $D$ be the average temporal maximum degree of $\mu$. If $G = (G_1, \dots, G_{n^2})$ where $G_1,\dots,G_{n^2}\sim \mu$ independently, then
    \[
        \Pr\big(\TEXP(G_1,\dots,G_{n^2})\leq C\Delta \cdot n \log^3 n\big)\geq 1-e^{-\log^3 n}.    
    \]
\end{corollary}
\begin{proof}
    Let $v_1,\dots,v_n$ be the order on the vertex set $[n]$ from Theorem~\ref{thm:main_proba2}. For $0\leq i \leq n-1$, let $T_0 \coloneqq 0$ and $T_i \coloneqq \min\{t\geq 1 : (G_{T_{i-1}+1},\dots,G_{T_i+t})\text{ contains a temporal walk from }v_i\text{ to }v_{i+1}\}$. Clearly, $\TEXP(G_1,\dots,G_{n^2})\leq \sum_i T_i$. Since the snapshots $G_t$ are independent, so are the random variables $T_{1},\dots,T_{n-1}$. Since the temporal graph $G$ is always-connected, we also have $T_{i}\leq n$, and therefore also $\sum \Ex[T_i^2]\leq n\sum \Ex[T_i]$. Since $\Ex[T_i] = \Ex_\mu[v_i \to v_{i+1}]$, this sum is at most $C\Delta \cdot n^2\log^3 n$ for some absolute constant $C \ge 3$. 

    By considering the random variables $T_i - \Ex[T_i]$ and using Bernstein's inequality (see e.g.~\cite[Theorem 2.9.5]{Vershynin_2018}), we obtain
    \begin{align*}
        \Pr\left(\sum_{i=1}^{n-1} T_{i} \geq 2C\Delta \cdot n \log^3 n\right) &\leq \exp\left\{-\frac{\frac{(C\Delta)^2}{2} n^2\log^6 n}{\sum_{i=1}^{n-1} \Ex[(T_i - \Ex[T_i])^2] + \frac{1}{3}C\Delta n^2\log^3 n}\right\} \\
        &\leq \exp\left\{-\frac{\frac{(C\Delta)^2}2 \log^3 n}{(4/3)C\Delta}\right\}.\\
    \end{align*}
    Since $C\geq 3$, this is at most $e^{-\log^3 n}$, which concludes the proof.
\end{proof}

\section{Temporal exploration for the Erd\h{o}s-Rényi random temporal graph}
\label{app:B}
A natural way to define the \textit{Erd\h{o}s-Rényi temporal random graph} is to sample each snapshot $G_t\sim G(n,p)$ independently for some value of $p=p(n)$. The snapshots will be connected with high probability when $p(n)\geq (1+\varepsilon)\log n/n$. We expect that there is a  threshold for having a temporal exploration of length $n-1$ at $p^*(n)=e/n$. Indeed, let
\[
    \H \coloneqq \{(v_1v_2,v_2v_3,\dots,v_{n-1}v_n):v_1,\dots,v_n \text{ is a permutation of $1,\dots,n$}\}
\]
denote the set of all (ordered) edges of Hamiltonian paths of $K_n$. 
We say that $(v_1v_2,\dots,v_{n-1}v_n)\in \mathcal{H}$ is \textit{contained} in the temporal graph $G$ if $v_iv_{i+1}$ is an edge of the $i$th snapshot of $G$.
The expected number of elements of $\mathcal{H}$ contained in $(G_1,\dots,G_{n-1})$ when $G_t\sim G(n,p)$ independently is
\[
n!p^{n-1} \sim \sqrt{2\pi n} \left(\frac{n}e\right)^{n}p^{n-1}.
\]
This number is exponentially small when $p(n)\leq (1-\varepsilon) e/n$ and exponentially large when $p(n)\geq (1+\varepsilon)e/n$ for $\varepsilon>0$. 

There is a nice way to show that for $p(n)\geq (1+\varepsilon)\log n/n$, there is a temporal exploration in $(G_1,\dots,G_{n-1})$ with high probability. This follows from the following simple application of a theorem of McDiarmid~\cite{mcdiarmid1980clutter}.

\begin{theorem}[Special case of Theorem 2.1 in \cite{mcdiarmid1980clutter}]
\label{thm:mcdiarmid}
    Let $n,k$ be positive integers and let $p\in(0,1]$. Let
    \[
        \mathcal{G}\subseteq\left\{(e_1,\dots,e_k) : e_i \in \binom{[n]}{2} \text{ and }e_i\neq e_j\text{ for all }i\neq j\right\}
    \]
    be a family of edge-ordered graphs, each with $k$ edges. Let $G_0,G_1,\ldots,G_k \sim G(n,p) $ be mutually independent. Then,
    \[
    \Pr\left((G_0,\ldots,G_0)\text{ contains an element of }\mathcal{G}\right)
    \leq
    \Pr\left((G_1,\ldots,G_k)\text{ contains an element of }\mathcal{G}\right),
    \]
    where $G_0$ is repeated $k$ times in the temporal graph on the left-hand side.
\end{theorem}

As before, we say that $(e_1,\dots,e_k)$ is contained in $(G_1,\dots,G_k)$ if $e_i\in E(G_i)$ for all $i\in [k]$. Now, let $p(n)$ be chosen so that $G_0\sim G(n,p)$ contains a Hamiltonian path with high probability. For example, it suffices to choose $p(n) \geq (1+\varepsilon)\log n/n$ for some $\varepsilon>0$.
Let $G=(G_1,\dots,G_{n-1})$ be the temporal graph where $G_t \sim G(n,p)$ independently. We claim that $G$ contains a temporal exploration (which necessarily has length $n-1$) with high probability. Indeed, by \cref{thm:mcdiarmid}, the probability that $G$ contains an element of $\mathcal{H}$ is at least as large as the probability that $(G_0,\dots,G_0)$ contains an element of $\mathcal{H}$, and this is exactly the probability that $G_0$ has a Hamiltonian path, which holds with high probability.

Erlebach, Hoffmann and Kammer~\cite[Theorem 6.1]{ErlebachHoffmannKammer21} studied a more general setting to the one above, in which each edge $e$ has its own edge probability $p_e$. If the total sum of probabilities of the edges over each cut of $G$ is at least a constant, they show that there exists a temporal exploration in $O(m\log n)$ steps for $m$ the number of edges and $n$ the number of vertices. Above we discussed the special case that $p_e=p$ for all $e$ and $m=\binom{n}2$, which is not an interesting case in their more general result.
\end{document}